\documentclass[lettersize,journal]{IEEEtran}
\usepackage{amsmath,amsfonts, amssymb}
\usepackage{array}
\usepackage[caption=false,font=footnotesize,labelfont=sf,textfont=sf]{subfig}
\usepackage{textcomp}
\usepackage{stfloats}
\usepackage{url}
\usepackage{verbatim}
\usepackage{graphicx}
\usepackage{cite}
\usepackage{algcompatible}
\usepackage{algorithm}
\usepackage{amsthm}
\usepackage{bbm}
\usepackage{orcidlink}
\newtheorem{theorem}{Theorem}

\newtheorem{definition}{Definition}
\newtheorem{corollary}{Corollary}
\newtheorem*{remark}{Remark}
\usepackage{makecell}
\usepackage[table]{xcolor}
\usepackage{graphicx}
\usepackage[nohyperlinks,nolist]{acronym}

\begin{document}

\title{QoS-Aware Load Balancing in the Computing Continuum via Multi-Player Bandits}

\author{
    \IEEEauthorblockN{
    Ivan Čilić \orcidlink{0009-0008-9330-1924}, \textit{Graduate Student Member, IEEE},
    Ivana Podnar Žarko \orcidlink{0000-0001-5619-2142}, \textit{Member, IEEE}, \\
    Pantelis A. Frangoudis \orcidlink{0000-0001-6901-7714}, 
    and Schahram Dustdar \orcidlink{0000-0001-6872-8821}, \textit{Fellow, IEEE}}

\thanks{This work has been supported in part by the European Union’s Horizon Europe research and innovation programme under grant agreement No. 101079214 (AIoTwin) and by the European Regional Development Fund (grant No. KK.01.1.1.01.0009, project DATACROSS). }% <-this % stops a space
\thanks{Ivan Čilić and Ivana Podnar Žarko are with the Faculty of Electrical Engineering and Computing, University of Zagreb, 10000 Zagreb, Croatia \{e-mail: ivan.cilic@fer.hr; ivana.podnar@fer.hr\}.}
\thanks{Pantelis A. Frangoudis and Schahram Dustdar are with the Distributed Systems Group, TU Wien, 1040 Vienna, Austria \{e-mail: p.frangoudis@dsg.tuwien.ac.at; dustdar@dsg.tuwien.ac.at\}.}
\thanks{Manuscript received April 19, 2021; revised August 16, 2021.}}

% The paper headers
\markboth{IEEE Transactions on Services Computing,~Vol.~X, No.~Y, Month~Year}%
{Ivan Čilić \MakeLowercase{\textit{et al.}}: QoS-Aware Load Balancing in the Computing Continuum via Multi-Player Bandits}

%\IEEEpubid{0000--0000/00\$00.00~\copyright~2021 IEEE}
% Remember, if you use this you must call \IEEEpubidadjcol in the second
% column for its text to clear the IEEEpubid mark.

\maketitle

\begin{abstract}
As computation shifts from the cloud to the edge to reduce processing latency and network traffic, the resulting Computing Continuum (CC) creates a dynamic environment where meeting strict Quality of Service (QoS) requirements and avoiding service instance overload becomes challenging. Existing methods often prioritize global metrics and overlook per-client QoS, which is crucial for latency-sensitive and reliability-critical applications.
We propose QEdgeProxy, a decentralized QoS-aware load balancer that acts as a proxy between IoT devices and service instances in the CC. We formulate the load balancing problem as a Multi-Player Multi-Armed Bandit (MP-MAB) with heterogeneous rewards: Each load balancer autonomously selects service instances to maximize the probability of meeting its clients’ QoS requirements by using Kernel Density Estimation (KDE) to estimate QoS success probabilities. Our load-balancing algorithm also incorporates an adaptive exploration mechanism to recover rapidly from performance shifts and non-stationary conditions.
We present a Kubernetes-native QEdgeProxy implementation and evaluate it on an emulated CC testbed deployed on a K3s cluster with realistic network conditions and a latency-sensitive edge-AI workload. Results show that QEdgeProxy significantly outperforms proximity-based and reinforcement-learning baselines in per-client QoS satisfaction, while adapting effectively to load surges and changes in instance availability.
\end{abstract}

\begin{IEEEkeywords}
edge computing, load balancing, Quality of Service (QoS), Multi-Player Multi-Armed Bandit (MP-MAB), Reinforcement Learning (RL), Kubernetes
\end{IEEEkeywords}

\section{Introduction}

The growing scale of Internet of Things (IoT) deployments has made routing all data to remote cloud servers for processing increasingly inefficient. This centralized model leads to increased latency, congestion, and limited responsiveness for real-time applications. The Computing Continuum (CC) addresses these challenges by distributing computation across IoT, edge, fog, and cloud nodes, enabling data to be processed closer to its source to improve both performance and resource utilization.
Despite these advantages, the CC remains a challenging environment as it is highly dynamic: nodes differ in resources, load, availability, and network conditions, while IoT workloads fluctuate over time and space. Ensuring that services deployed across this heterogeneous continuum continue to meet strict Quality of Service (QoS) requirements therefore demands mechanisms that can adapt to changing conditions.

Reliable service delivery in the CC depends on the explicit definition and enforcement of QoS requirements. These requirements specify measurable performance guarantees, typically expressed through Service Level Objectives (SLOs) and Service Level Agreements (SLAs)~\cite{wsla}. SLOs define concrete targets, such as ensuring that a given percentage of requests remain below a latency threshold, while SLAs formalize these guarantees contractually between providers and consumers.
QoS requirements in the CC may be expressed along multiple dimensions, including latency, throughput, reliability, availability, privacy, and security. Their importance varies across application domains; for example, industrial control and autonomous systems prioritize ultra-low latency, whereas healthcare applications may emphasize both privacy and availability. Formalizing such QoS constraints provides a foundation for coordinating placement, routing, and adaptation strategies so that services remain within their expected performance boundaries under dynamic and heterogeneous CC conditions.

This work addresses two key challenges that arise once service instances are deployed in the CC across different layers of the continuum: (i) guaranteeing QoS for continuous IoT data processing for each client through adaptive, decentralized request routing under uncertainty, and (ii) maintaining continuous data delivery from multiple independent clients without overloading service instances on resource-constrained devices or compute nodes.

To support QoS-aware routing in dynamic CC conditions, our previous work~\cite{QEdgeProxy} introduced \textit{QEdgeProxy}, a lightweight decentralized load balancer (LB) deployed on every CC node. QEdgeProxy acts as an intermediary between clients and service instances, selecting a suitable instance for each client request based on locally observed QoS parameters while continuously adapting its routing decisions as conditions in the CC change. Rather than selecting a single optimal instance that maximizes QoS for each client request, it maintains a pool of instances (\textit{QoS pool}) capable of meeting the QoS requirements with high probability. This approach enables effective load balancing while preserving per-client QoS guarantees. QEdgeProxy is designed to handle both client-to-service and service-to-service interactions, allowing IoT devices and intermediate services to forward requests through a QEdgeProxy instance deployed on their nearest or hosting node. By maintaining and dynamically updating a QoS pool, QEdgeProxy ensures balanced request distribution and mitigates instance overload while satisfying QoS requirements for each client.

While the high-level architectural design and preliminary proof-of-concept implementation of QEdgeProxy were introduced in~\cite{QEdgeProxy}, the load-balancing logic was limited to constructing QoS pools based on average latency measurements and distributing requests within the pool using simple round-robin scheduling. This paper significantly extends and formalizes the load-balancing mechanism by developing the underlying algorithms for adaptive QoS-pool maintenance and request routing. We formulate the problem as a Multi-Player Multi-Armed Bandit (MP-MAB) with heterogeneous rewards~\cite{slivkins2019} and provide a decentralized learning procedure—built on Kernel Density Estimation (KDE)-based QoS estimation, adaptive exploration, and instance-specific weight updates—that enables QEdgeProxy to operate robustly under dynamic and non-stationary CC conditions. This work provides the theoretical foundation, detailed algorithmic implementation, and a comprehensive evaluation that were not addressed in the initial QEdgeProxy prototype, resulting in a complete and rigorous load-balancing solution.

Unlike most prior work that prioritizes global performance metrics, such as minimizing average latency or total delay~\cite{proximity,wang2024decentralized}, or relies on partial coordination among decision makers~\cite{AdeQoS,Dec-SARSA}, this paper focuses on guaranteeing \emph{per-client} QoS requirements through fully decentralized load balancing. By formulating routing as an MP-MAB with heterogeneous and non-stationary rewards, our approach enables each LB to adapt routing decisions based solely on local QoS feedback, without centralized control or inter-node coordination.

The main contributions can be summarized as follows:

\begin{itemize}
    \item \textbf{MP-MAB problem formulation for QoS-aware load balancing.}
    We formulate the load-balancing problem in CC as an MP-MAB with heterogeneous and non-stationary rewards. Unlike classical MAB or offloading models based on single-action selection, our formulation captures weighted request routing, shared service instances across multiple LBs, and strict per-client QoS requirements.

    \item \textbf{Decentralized algorithms for QoS-aware load balancing in CC.}
    We develop a decentralized set of algorithms that solve the proposed MP-MAB formulation and enable QoS-aware load balancing in dynamic CC environments.
    Each instance of LB autonomously maintains a \textit{QoS pool}---a set of service instances predicted to meet the required QoS---using a combination of $\varepsilon$-decay exploration and KDE to estimate instance-level QoS success probabilities. Furthermore, it incorporates an adaptive exploration mechanism that reduces exploration as confidence grows, but increases it when QoS degradation or performance shifts are detected, enabling rapid recovery in non-stationary conditions.
    Requests are routed using a Smooth Weighted Round Robin (SWRR) mechanism parameterized by these instance-level probabilities, ensuring balanced utilization of service instances while preserving per-client QoS guarantees.

    \item \textbf{Kubernetes-native implementation.}
    We implement the algorithms within a Kubernetes-native component (QEdgeProxy), enabling each node to run a decentralized QoS-aware load balancer that maintains QoS pools, performs adaptive weighted routing, and obtains node and instance updates through an MQTT-based informer to reduce the monitoring overhead. This design demonstrates that decentralized QoS-aware load balancing can be seamlessly integrated into existing CC orchestration solutions with minimal configuration.

    \item \textbf{Real-world validation in a distributed CC testbed.}
    We evaluate QEdgeProxy instances deployed across a real K3s cluster emulating a 30-node CC topology with realistic network latencies. 
    Using a latency-sensitive IoT edge-AI workload, we show that our solution outperforms proximity-based and reinforcement-learning baselines, achieving higher per-client QoS satisfaction, improved fairness, robust adaptability to dynamic changes, and low resource overhead suitable for constrained edge devices.
\end{itemize}

The remainder of this paper is organized as follows. 
Section~\ref{section:related} reviews related work on QoS-aware routing, decentralized load balancing, and learning-based decision-making in the CC. 
Section~\ref{section:model} introduces the system model, describing both design-time and runtime behavior of the system, and QoS requirements definition. 
Section~\ref{section:problem} formalizes the decentralized QoS-aware load balancing problem as an MP-MAB problem with heterogeneous rewards. 
Section~\ref{section:algorithms} proposes the decentralized load-balancing algorithms implemented in QEdgeProxy, including QoS-pool construction, KDE-based estimation, $\varepsilon$-decay exploration, and SWRR-based request routing. 
Section~\ref{section:impl-qedgeproxy} details the Kubernetes-native implementation of the QEdgeProxy, and Section~\ref{section:evaluation} reports the experimental evaluation conducted on an emulated CC testbed. 
Finally, Section~\ref{section:conclusion} concludes the paper. In Table~\ref{tab:notation} we listed all the symbols used in the paper.

\begin{table}[ht]
\centering
\caption{List of symbols used.}
\label{tab:notation}
\renewcommand{\arraystretch}{1.2}
\begin{tabular}{ll}
\hline
\textbf{Symbol} & \textbf{Description} \\
\hline
\multicolumn{2}{c}{Model} \\[-1pt]
\hline
$\mathcal{K}$ & Set of load balancers (LB); one per CC node \\
$S$ & IoT service deployed in the CC \\ [2pt]
$(\rho, \tau, W)$ & \makecell[l]{QoS requirements: required success ratio $\rho$ \\and latency threshold $\tau$ evaluated over \\window $W$} \\
$\mathcal{M}(t)$ & Active service instances \\
$l_{k,m}^n(t)$ & Network RTT between LB $k$ and instance $m$\\
$l_m^p(t)$ & Processing time of instance $m$ \\ [2pt]
$\ell_{k,m}(t)$ & \makecell[l]{End-to-end latency for a request from LB $k$ \\processed by instance $m$} \\
$\lambda_k(t)$ & Request arrival rate at LB $k$ \\
$\Lambda_m(t)$ & Request arrival rate at instance $m$ \\
$w_{k,m}(t)$ & Routing weight assigned by LB $k$ to instance $m$ \\
$r_{k,m}(t)$ & Reward for LB $k$ routing a request to $m$ \\ [2pt]
$\mathcal{R}_k(t;W)$ & \makecell[l]{Set of requests handled by load balancer $k$ in the \\window $[t{-}W,\,t)$} \\ [2pt]
$\mu_{k,m}(t)$ & \makecell[l]{Expected QoS success probability for LB $k$ on \\instance $m$} \\
\hline
\multicolumn{2}{c}{Algorithms} \\[-1pt]
\hline
$\varepsilon(t)$ & Exploration factor \\
$\gamma$ & $\varepsilon$-decay factor \\
$\eta$ & Score smoothing factor \\
$\mathcal{F}_k(t)$ & Set of feasible instances \\
$\mathcal{E}_k(t)$ & Exploitation pool \\
$\mathcal{X}_k(t)$ & Exploration pool \\
$\mathcal{Q}_k(t)$ & QoS pool \\
$\hat{\rho}_k(t)$ & Empirical QoS success ratio at LB $k$ over window $W$ \\
$E_t$ & Error count threshold \\ 
$\Delta_{\mathrm{cd}}$ & Cooldown duration \\
\hline
\multicolumn{2}{l}{* Time-dependent quantities are indexed by discrete time $t$.}
\end{tabular}
\vspace{-0.2cm}
\end{table}

\section{Related work}
\label{section:related}

Adaptive and QoS-aware routing in the CC has been investigated from several perspectives, including proximity-aware routing, decentralized multi-agent scheduling, and QoS-aware decision making based on multi-armed bandit (MAB) formulations. This section reviews the most relevant approaches and discusses their relation to our proposed load-balancing mechanism.

Fahs and Pierre introduced \textit{Proxy-mity}~\cite{proximity}, a proximity-aware traffic routing system designed as a plugin for Kubernetes. Their goal was to make Kubernetes' routing aware of geographical distribution by enabling a trade-off between load balancing and latency reduction through a controllable parameter~$\alpha$. When $\alpha=0$, traffic is evenly distributed across all service replicas, whereas $\alpha=1$ prioritizes routing to the closest instance. The approach seamlessly integrates with Kubernetes and reduces the request latency by up to 90\% in experimental fog setups. However, Proxy-mity considers only network latency as a proximity metric and does not incorporate feedback from actual request processing times. Moreover, it does not allow the explicit definition of QoS constraints such as maximum latency thresholds or success ratios. 

Liu and Fang~\cite{AdeQoS} proposed \textit{AdeQoS}, a decentralized multi-player multi-armed bandit (MP-MAB) algorithm that guarantees per-player QoS satisfaction. Their model assumes that each player (e.g., a load balancer) interacts with a set of arms (service instances) without coordination or observation of others' actions. Unlike conventional reward maximization, AdeQoS tries to ensure that each player maintains an expected reward above a given QoS threshold, achieving $O(1)$ QoS regret. This approach is closely aligned with our formulation (see Section~\ref{section:problem}), as both aim to meet per-client QoS constraints rather than maximizing aggregate utility. However, AdeQoS assumes that the number of arms exceeds the number of players and that collisions occur when two players select the same arm, which differs from our model (Section~\ref{section:model}) where each service instance can handle multiple requests simultaneously. Moreover, AdeQoS relies on a leader-based coordination phase to synchronize weights, whereas our solution is fully decentralized, with each load balancer independently adapting weights through local QoS measurements. Nonetheless, both methods share the key objective of distributed QoS maintenance without centralized orchestration.

Wang \emph{et al.} \cite{wang2024decentralized} proposed a decentralized task offloading framework for edge computing formulated as an MAB problem. Their approach enables multiple users to offload computational tasks to nearby edge servers through an auction-like interaction where each user selects one server per round to minimize its own latency. While the algorithm supports distributed decision-making and achieves sublinear regret, it assumes that all agents share information via a coordinating leader that announces the winning allocations after each auction round. This requirement of centralized synchronization limits applicability in fully decentralized CC environments. Moreover, each user exploits only a single server following the auction, which constrains scalability for load balancing where multiple clients concurrently access shared servers. 

Zheng \emph{et al.} \cite{tahir2025computation} extended this research by introducing an auction-based computation-offloading framework under incomplete information. Their system couples distributed MAB learning with a heterogeneous auction mechanism to adapt to dynamic network and load variations. The \textit{TODAuction-MAB} algorithm uses the estimated rewards to decide how to allocate servers to each user through an auction-based mechanism, while the \textit{HD-MAB} algorithm adds sliding-window updates to handle non-stationary conditions. Despite strong performance and theoretical guarantees, their architecture still relies on a leader node to coordinate distributed auctions and assumes that each server handles at most one user at a time, with the total number of servers exceeding the number of users. These assumptions simplify the offloading model but limit generalization to realistic CC scenarios.

Mattia and Beraldi \cite{Dec-SARSA} proposed \textit{Dec-SARSA}, a differential SARSA\footnote{SARSA (State–Action–Reward–State–Action) is an on-policy reinforcement learning algorithm that updates action values based on the current transition and the policy’s own next action, enabling agents to learn optimal behavior through direct interaction with the environment~\cite{rl-intro}.} reinforcement-learning (RL) algorithm for online decentralized task scheduling in fog computing. Each fog node operates as an autonomous agent that learns to offload or execute tasks locally to maximize the fraction of tasks meeting their deadlines. The algorithm learns online using only local state and reward feedback, without centralized coordination, and has been validated both in simulation and in a real deployment. Similar to our design, they define QoS satisfaction as meeting a latency threshold and compute the total latency as the sum of network delay and processing (queuing) time. However, before deciding to offload a task, a node must explicitly query neighboring nodes to obtain their current load, which limits scalability and practicality in environments where such inter-node communication is not guaranteed.

Most of the reviewed approaches primarily focus on minimizing average latency or total delay~\cite{wang2024decentralized,tahir2025computation,brun2012worstcase,yu2023iotmicroservices,zhu2023dynamicstream,mattia2023lb,ASLANPOUR2024266}, but they rarely formulate QoS as a strict per-client constraint defined by a latency threshold~\cite{AdeQoS,Dec-SARSA,nezami2021decentralizedlb}. Consequently, they aim for global performance improvement rather than ensuring that each client consistently meets its QoS target. Furthermore, these works rarely consider a realistic CC deployment scenario in which service instances are already pre-deployed, e.g., by an orchestration framework, and can simultaneously serve multiple clients. Among recent studies, there is a clear methodological trend toward modeling routing and offloading as MAB~\cite{wang2024decentralized,tahir2025computation,yahya2024decoffloading} or RL~\cite{Dec-SARSA} problems, enabling adaptive decision making under uncertainty. However, many of these methods still rely on partial centralization, such as leader-based coordination in auction mechanisms or explicit inter-node communication for load estimation, and thus fail to achieve full decentralization.  

This state-of-the-art review has motivated us to formulate our routing problem as a multi-player multi-armed bandit, enabling adaptive learning of instance selection under uncertainty while ensuring decentralized operation (see Section~\ref{section:problem}). Instead of relying on one-to-one offloading decisions or coordinated auctions, our approach introduces a decentralized load balancer that continuously adjusts routing weights based on locally observed QoS feedback. By combining the adaptive exploration–exploitation principles of the MAB framework with dynamic weight adaptation, the proposed \textbf{QEdgeProxy} achieves decentralized and scalable routing that maintains per-client QoS satisfaction and balances load across multiple active instances in real time.

\section{System Model}
\label{section:model}

We consider a decentralized system consisting of:
\begin{itemize}
    \item $K$ load balancers, indexed by $k \in \mathcal{K} = \{1, \dots, K\}$.
    \item Service $S$ associated with a set of instances~$\mathcal{M} = \{1, \dots, M\}$. 
\end{itemize}

Each LB $k$ serves a client request stream at a known rate $\lambda_k \geq 0$, which may vary across LBs. Each client is associated with a single LB responsible for routing its traffic. 
For each LB $k$, we denote by
\begin{equation}
    w_k(t) = \big(w_{k,1}(t), \ldots, w_{k,M(t)}(t)\big),
\end{equation}
\begin{equation}
    \sum_{m \in \mathcal{M}(t)} w_{k,m}(t) = 1
\end{equation}
the \emph{routing weight vector} at time~$t$, where $w_{k,m}(t)$ represents the fraction of traffic that LB~$k$ forwards to instance~$m$. 

Each service is characterized by its resource requirements (CPU, memory, and network bandwidth) and the QoS requirements~$(\tau, \rho, W)$\footnote{Following our related work overview in~\cite{cilicPerf}, latency is selected as the QoS parameter, as it is the most commonly used metric.}, where:
\begin{itemize}
    \item $\tau$ is the maximum acceptable end-to-end latency per request, i.e., the request–response round-trip time.
    \item $\rho \in (0, 1]$ is the minimum required fraction of requests that must meet the latency threshold $\tau$. 
    \item $W$ is the time window (in time units) over which the QoS requirements are evaluated.
\end{itemize}
The QoS requirements must be satisfied for each client within a sliding time window of duration~$W$, such that at least a fraction~$\rho$ of the client's requests are served within the latency threshold~$\tau$.

We focus on a single service~$S$ to simplify the problem formulation and analysis. In practical CC deployments, multiple services may coexist, each with its own set of instances and QoS requirements. In such cases, the proposed model and routing mechanisms are applied independently per service, resulting in separate routing weight vectors for each service. 

During runtime, the following parameters may vary:
\begin{itemize}
    \item $l_{k,m}^{n}(t)$: \underline{\textbf{n}}etwork latency between LB $k$ and instance $m$ (round-trip).
    \item $l_m^{p}(t)$: request \underline{\textbf{p}}rocessing latency of instance $m$, including queuing and execution delay.
    \item $\lambda_k(t)$: incoming request rate at LB $k$.
    \item $\Lambda_m(t)$: incoming request rate at instance $m$.
\end{itemize}

From these parameters, the end-to-end latency for a request from LB $k$ processed by instance $m$ at time $t$ is given by:
\begin{equation}
\ell_{k,m}(t) = l_{k,m}^n(t) + l_m^p(t),
\end{equation}
which determines whether the QoS constraint $\ell_{k,m}(t) \leq \tau$ is satisfied or not.

Fig.~\ref{fig:decentralized-LB}\footnote{The figure was inspired by a similar example presented in~\cite{brun2012worstcase}.} illustrates the decentralized load-balancing model considered in this work. Each LB $k \in \mathcal{K}$ receives an aggregate incoming request rate $\lambda_k(t)$ from its assigned clients and routes these requests toward the active service instances $m \in \mathcal{M}(t)$ based on adaptive routing weights $w_{k,m}(t)$. The weights determine the fraction of traffic that LB~$k$ forwards to each instance, while taking into account network conditions and the observed QoS performance.

\begin{figure}[ht]
\centering
\includegraphics[width=0.9\linewidth]{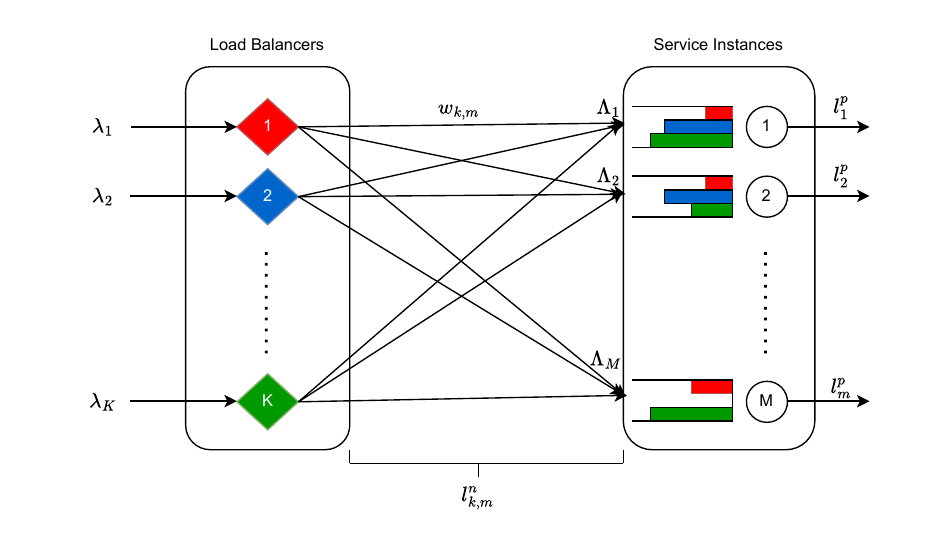}
\caption{Decentralized load balancing.}
\label{fig:decentralized-LB}
\end{figure}

\section{MP-MAB Problem Formulation}
\label{section:problem}

The runtime load-balancing problem is modeled as a \textbf{multi-player multi-armed bandit (MP-MAB)} process with \textbf{heterogeneous rewards}, where multiple load balancers (players) concurrently select service instances (arms) to route their incoming requests.

\begin{itemize}
    \item Each load balancer (LB) $k \in \mathcal{K}$ maintains a probability distribution (weight vector) $w_k(t)$ which determines the fraction of its incoming request rate $\lambda_k(t)$ that is routed to each active instance $m \in \mathcal{M}(t)$. We assume that $K \ge M(t)$ holds at any time $t$, reflecting the common CC configuration where service instances are deployed only on a subset of available nodes and can serve multiple requests simultaneously. By design, our system deploys one LB per node.
    \item The aggregate incoming request rate at instance $m$ is
    \begin{equation}
    \Lambda_m(t) = \sum_{k \in \mathcal{K}} w_{k,m}(t)\,\lambda_k(t).
    \end{equation}
    \item The instantaneous reward for LB $k$ when it routes a request to instance $m$ is defined as
    \begin{equation}
    r_{k,m}(t) =
    \begin{cases}
        1, & \text{if } \ell_{k,m}(t) \le \tau, \\
        0, & \text{otherwise,}
    \end{cases}
    \end{equation}
    where $\ell_{k,m}(t)$ is the end-to-end latency observed by LB $k$ through instance $m$ at time $t$.
    \item The mean reward is $\mu_{k,m} = \mathbb{P}(\ell_{k,m}(t) \le \tau)$, which is:
    \begin{itemize}
        \item \textbf{heterogeneous} across LBs due to different network latencies $l_{k,m}^n(t)$, and
        \item \textbf{time-varying} due to changes in processing latency $l_m^p(t)$ caused by collisions.
    \end{itemize}
\end{itemize}

\subsection{Assumptions}
We formulate the load-balancing problem under the following assumptions: 
\begin{itemize}
    \item Each LB estimates the network latencies $l_{k,m}^n(t)$ towards the available instances through lightweight periodic probes whose overhead is negligible compared to the service traffic.
    \item Client-to-LB latencies are assumed to be negligible and absorbed into $\ell_{k,m}(t)$, reflecting typical edge deployments where ingress gateways are located in close proximity to end devices.
    \item Processing latency $l_m^p(t)$ increases consistently with $\Lambda_m(t)$.
    \item Each LB observes only its own latency samples $\ell_{k,m}(t)$ and rewards $r_{k,m}(t)$, without access to other LBs’ weights or observations. Coordinated decision-making would require each LB to obtain performance or state information from others, which would introduce significant communication and monitoring overhead. This is particularly impractical in large-scale CC deployments with thousands of geographically distributed nodes.
    \item LBs are aware of instance placement changes, i.e., when new instances are deployed or existing ones are removed. To limit the monitoring overhead, each LB receives such updates only for instances hosted on nodes within its potential QoS reach (nodes that could satisfy the latency threshold~$\tau$ under current network conditions).
    \item The system is well-provisioned: there exists a centralized oracle allocation of weights that satisfies the QoS threshold $\tau$ for all clients.
\end{itemize}

\subsection{Objective}

Each LB~$k$ must learn a decentralized policy~$\pi_k$ that adaptively updates its weight vector~$w_k(t)$, without coordination with other LBs, so that, for every time~$t$,
\begin{equation}
\label{eq:objective}
\frac{1}{|\mathcal{R}_k(t;W)|}\sum_{j \in \mathcal{R}_k(t;W)} r_{k,j} \;\ge\; \rho,
\end{equation}
where $\mathcal{R}_k(t;W)$ is the set of requests handled by load balancer~$k$ in the window~$[t{-}W,\,t)$, and $r_{k,j}$ indicates whether request~$j$ met the latency target~$\tau$.  
The success indicator~$r_{k,j}$ is determined by the instance selected for that request,
\begin{equation}
r_{k,j} = r_{k,m}(t) \quad \text{where} \quad m = a_{k,j},
\end{equation}
with $a_{k,j}$ denoting the instance chosen by LB~$k$ for request~$j$.

\subsection{Key Properties}

\begin{itemize}
    \item \textbf{Players:} $K$ load balancers, one per node ($k \in \mathcal{K}$).
    \item \textbf{Arms:} Active service instances $m \in \mathcal{M}(t)$.
    \item \textbf{Adaptive weights:} Each policy outputs a probability distribution over arms, which can change during runtime.
    \item \textbf{Heterogeneous rewards:} Rewards are heterogeneous due to per-LB network latency $l_{k,m}^n$.
    \item \textbf{Implicit collisions:} When multiple LBs select the same instance, congestion increases its processing latency $l_m^p(t)$, potentially reducing rewards.
    \item \textbf{QoS-driven:} Each LB maximizes the probability of satisfying QoS requirements $(\rho,\tau, W)$ rather than minimizing overall mean latency.
    \item \textbf{Decentralized:} 
    LBs operate independently, without sharing information or explicit coordination. 
\end{itemize}

\subsection{Regret Definition}

Regret is the standard performance measure in MAB problems. It quantifies the cumulative loss incurred by a learning algorithm compared to an oracle that always selects the optimal action with full knowledge of the underlying reward distributions \cite{slivkins2019}. A low (sublinear) regret indicates that the learning policy efficiently balances exploration and exploitation, ultimately achieving performance comparable to that of the oracle. In the context of this work, each LB acts as a learning agent that adaptively adjusts its routing weights to maximize QoS success. Therefore, regret represents the cumulative QoS performance gap between the decentralized routing policy and an oracle weight allocation that knows the success probabilities of all service instances.

Let:
\begin{itemize}
    \item $\mu_{k,m}(t) = \mathbb{P}(\ell_{k,m}(t) \le \tau)$ denote the expected QoS success probability for LB $k$ on instance $m$ at time $t$,
    \item $w_k^*(t) \;=\; \arg\max_{w} \;\sum_{m \in \mathcal{M}(t)} w_m\, \mu_{k,m}(t)$ be the oracle weight vector at time $t$,
    \item $w_k(t)$ be the learned weight vector at time $t$.
\end{itemize}

Then, the cumulative regret of LB $k$ over $T$ rounds is
\begin{equation}
\begin{split}
R_k(T) = \sum_{t=1}^{T} \Bigg( 
\sum_{m \in \mathcal{M}(t)} w_{k,m}^*(t)\,\mu_{k,m}(t)
\;-\; \\
\sum_{m \in \mathcal{M}(t)} w_{k,m}(t)\,\mu_{k,m}(t)
\Bigg),
\end{split}
\end{equation}
and the total system regret is
\begin{equation}
\label{eq:overall-regret}
R(T) = \sum_{k \in \mathcal{K}} R_k(T).
\end{equation}

\section{Decentralized Algorithms for QoS-Aware Load Balancing in CC}
\label{section:algorithms}

We propose \textbf{a set of decentralized algorithms for QoS-aware load balancing in the CC}. Each LB operates as an independent learner, continuously adapting its routing weights toward the set of service instances that can meet its clients’ QoS requirements. The algorithms are designed to be fully decentralized and lightweight, relying only on local latency measurements and limited placement updates from the orchestrator. This enables scalability across large CC deployments with thousands of nodes and dynamic instance placement.

Each LB operates in continuous time indexed by discrete units $t = 1, 2, 3, \dots$. The adaptation process is organized into discrete update intervals indexed by $d = 1, 2, 3, \dots$, with each step spanning a fixed horizon of $H_d$ time units.  At the end of each decision step, an LB updates its estimates of instance QoS performance, recalculates routing weights, and adjusts its exploration budget. LBs receive placement updates (i.e., instance additions or removals) only for nodes within their potential QoS reach (nodes that can satisfy the latency threshold~$\tau$) to minimize monitoring and communication overhead.

\subsection{QoS Pool Maintenance}

Algorithm~\ref{alg:qos-pool-weights} describes the main adaptive control loop executed by each LB, which maintains the service QoS pools. At the beginning of every decision step, the LB updates its local view of the available instances by applying placement updates received from the orchestrator (line \ref{line:feasible}). The feasible set $\mathcal{F}_k(t)$ includes all instances within the latency bound~$\tau$ that are currently not in cooldown (more on cooldown in  Algorithm~\ref{alg:routing-weights}). Each instance $m \in \mathcal{F}_k(t)$ is then evaluated using empirical QoS estimates $\hat{\mu}_{k,m}(t)$ obtained with Kernel Density Estimation (KDE),\footnote{Kernel Density Estimation (KDE) is a non-parametric technique for estimating the probability density function of an unknown distribution from observed samples. Rather than producing a blocky histogram, KDE places smooth kernel functions at each data point and aggregates them to form a continuous density estimate~\cite{kde}.} computed over a sliding observation window of $W$ time units, representing the probability that latency $\ell_{k,m}(t)$ remains below the target~$\tau$ (line \ref{line:estimate}).

Instances with an estimated QoS success probability~$\hat{\mu}_{k,m}(t)$ exceeding the target success ratio~$\rho$ are assigned to the \emph{exploitation pool}~$\mathcal{E}_k(t)$, while the remaining feasible instances from $\mathcal{F}_k(t)$ form the \emph{exploration pool}~$\mathcal{X}_k(t)$. The overall QoS pool $\mathcal{Q}_k(t)$ is defined as their union (line~\ref{line:qos-pool}). 

Routing weights~$w_{k,m}(t)$ are computed proportionally to QoS scores within each pool. The total weight is divided between the exploration and exploitation pools according to the exploration rate~$\varepsilon(t)$ (line~\ref{line:budgets}). Then, per-pool scores are derived: exploitation scores increase with the excess of~$\hat{\mu}_{k,m}(t)$ over the QoS target~$\rho$, while exploration scores increase as~$\hat{\mu}_{k,m}(t)$ approaches~$\rho$ (line~\ref{line:scores}). To prevent zero-valued scores and ensure stable normalization, a small smoothing factor~$\eta$ is added to each score before normalization.
These scores are normalized within their respective pools, producing the final routing weights~$w_k(t)$ (line~\ref{line:final-weights}), with instances outside the QoS pool assigned zero weight.
As the measured QoS success ratio $\hat{\rho}_k(t)$ stabilizes, the exploration rate $\varepsilon(t)$ decays (line \ref{line:decay-epsilon}); if degradation is detected, exploration is reset to adapt to the current system state (line \ref{line:reset-epsilon}).

\begin{algorithm}[]
\caption{QoS pool maintenance}
\label{alg:qos-pool-weights}
\begin{algorithmic}[1]
\REQUIRE Threshold $\tau$, QoS target $\rho$, window $W$, $\varepsilon$-decay factor $\gamma$, smoothing factor $\eta>0$
\STATE \textbf{Initialization at $t{=}0$:}
\STATE Receive $\mathcal{M}(0)$; define feasible set by QoS reach:
\[
\mathcal{F}_k(0) \gets \{\, m \in \mathcal{M}(0) \mid l_{k,m}^n \le \tau\}
\] \label{line:feasible} \vspace{-0.5cm}
\STATE Set exploitation pool $\mathcal{E}_k(0) \gets \emptyset$ ,\quad exploration pool $\mathcal{X}_k(0) \gets \mathcal{F}_k(0)$
\STATE Define QoS pool $\mathcal{Q}_k(0) \gets \mathcal{E}_k(0) \cup \mathcal{X}_k(0)$
\STATE Set exploration budget $\varepsilon(0) \gets 1 - \rho$; initialize uniform $w_{k,m}(0)$ over $\mathcal{Q}_k(0)$

\FOR{each decision step $d = 1, 2, \dots$ (where $t = d \times H_d$)}
    \STATE Refresh: network latencies $l_{k,m}^n$, instances $\mathcal{M}(t)$, and cooldown states $T^{\mathrm{cd}}_m$
    \STATE $l^{p\star} \gets$ best expected processing latency
    \STATE Update feasible set:
    \[
    \mathcal{F}_k(t) \gets 
    \{\, m \in \mathcal{M}(t) 
    \mid l_{k,m}^n{+}l^{p\star} \le \tau \ \wedge\ 
    t \ge T^{\mathrm{cd}}_m \,\}
    \]
    \STATE Set $\mathcal{E}_k(t) \gets \emptyset$, $\mathcal{X}_k(t) \gets \emptyset$
    \FOR{each $m \in \mathcal{F}_k(t)$}
        \STATE Estimate QoS success over the last window $W$:
        \[
        \hat{\mu}_{k,m}(t) \gets \mathbb{P}\!\big(\ell_{k,m}(t') \le \tau \ \big|\ t' \in [t{-}W, t)\big)
        \] \label{line:estimate} \vspace{-0.5cm}
        \IF{$\hat{\mu}_{k,m}(t) \ge \rho$} \STATE add $m$ to $\mathcal{E}_k(t)$
        \ELSE \STATE add $m$ to $\mathcal{X}_k(t)$
        \ENDIF
    \ENDFOR
    \STATE Define QoS pool $\mathcal{Q}_k(t) \gets \mathcal{E}_k(t) \cup \mathcal{X}_k(t)$ \label{line:qos-pool}
    \STATE Allocate pool budgets: 
        \[W^{\mathrm{E}}(t) \gets 1-\varepsilon(t), W^{\mathrm{X}}(t) \gets \varepsilon(t)\] \label{line:budgets} \vspace{-0.5cm}
    \STATE Compute per-pool scores (with floor $\eta$ to avoid zeros):
    \[
    s^{\mathrm{E}}_{k,m}(t) = (\hat{\mu}_{k,m}(t)-\rho) + \eta \;\; (m\in\mathcal{E}_k(t)), 
    \]
    \[
    s^{\mathrm{X}}_{k,m}(t) = \hat{\mu}_{k,m}(t) + \eta \;\; (m\in\mathcal{X}_k(t))
    \] \label{line:scores} \vspace{-0.5cm}
    \STATE Normalize scores within pool budgets:
    \[
    w^{\mathrm{E}}_{k,m}(t) = 
      W^{\mathrm{E}}(t)\,\frac{s^{\mathrm{E}}_{k,m}(t)}{\sum_{j \in \mathcal{E}_k(t)} s^{\mathrm{E}}_{k,j}(t)} ,
    \]
    \[ 
    w^{\mathrm{X}}_{k,m}(t) = 
      W^{\mathrm{X}}(t)\,\frac{s^{\mathrm{X}}_{k,m}(t)}{\sum_{j \in \mathcal{X}_k(t)} s^{\mathrm{X}}_{k,j}(t)}
    \] \vspace{-0.3cm}
    \STATE Form final weights:
    \[
    w_{k,m}(t) \gets 
      \begin{cases}
        w^{\mathrm{E}}_{k,m}(t) + w^{\mathrm{X}}_{k,m}(t), & m \in \mathcal{Q}_k(t),\\
        0, & m \notin \mathcal{Q}_k(t)
      \end{cases}
    \] \label{line:final-weights}
    \IF{$\hat{\rho}_k(t) < \hat{\rho}_k(t-W)$} 
        \STATE reset $\varepsilon(t{+}1) \gets 1 - \rho$ \label{line:reset-epsilon}
    \ELSE
        \STATE decay $\varepsilon(t{+}1) \gets \varepsilon(t)\cdot \gamma$ \label{line:decay-epsilon}
    \ENDIF
\ENDFOR
\end{algorithmic}
\end{algorithm}

This mechanism allows each LB to continuously rebalance weights toward instances with higher QoS success probability while preserving a nonzero exploration rate that enables discovery of newly added or performance-improving instances. The computation is local, lightweight, and scales well with the number of reachable instances.

\subsection{Request Routing}

Algorithm~\ref{alg:routing-weights} defines the per-request routing process that operates continuously in parallel with weight maintenance. Upon receiving a request, each LB selects a target instance $m$ from its current QoS pool $\mathcal{Q}_k(t)$ using the Smooth Weighted Round Robin (SWRR)\footnote{SWRR is an implementation of the round robin mechanism, introduced by the NGINX HTTP load balancer \cite{nginx} and designed to smooth the distribution of requests over time, preventing high-weight targets from receiving large bursts of requests while others wait excessively to receive new requests.} policy based on the weight vector $w_k(t)$. The request is routed according to the weights, while the observed latency $\ell_{k,m}(t)$ and QoS success $r_{k,m}(t)$ are stored as feedback for the next decision window (line~\ref{line:route-req}). Instances that repeatedly violate the latency threshold~$\tau$ exceeding a predefined error count $E_t$ are temporarily placed in a cooldown state for a fixed duration $\Delta_{\mathrm{cd}}$ and removed from the QoS pool, ensuring that degraded instances do not continue to receive traffic (line~\ref{line:set-cooldown}). This continuous feedback loop allows each LB to adapt routing decisions dynamically to current network and service conditions without any inter-LB coordination or shared global state.

\begin{algorithm}[]
\caption{Request routing using weights}
\label{alg:routing-weights}
\begin{algorithmic}[1]
\REQUIRE Weights $w_{k,m}(t)$, QoS threshold $\tau$, error count threshold $E_t$, cooldown duration $\Delta_{\mathrm{cd}}$
\FOR{each request $r = 1, 2, \dots$}
    \STATE $t \gets$ current time
    \STATE Select instance $m \in \mathcal{Q}_k(t)$ using SWRR with $w_k(t)$
    \STATE Route request; record latency $\ell_{k,m}(t)$ and reward $r_{k,m}(t)=\mathbf{1}\{\ell_{k,m}(t)\le \tau\}$ \label{line:route-req}
    \STATE Update $e_{k,m} \leftarrow 0$ if $r_{k,m}(t)=1$ else $e_{k,m} \leftarrow e_{k,m}+1$
    \IF{$e_{k,m} \ge E_t$}
        \STATE Set $T^{\mathrm{cd}}_m \leftarrow t + \Delta_{\mathrm{cd}}$; remove $m$ from $\mathcal{Q}_k(t)$; renormalize $w_{k}(t)$ \label{line:set-cooldown}
        \STATE Set $e_{k,m} \leftarrow 0$
    \ENDIF
\ENDFOR
\end{algorithmic}
\end{algorithm}

\subsection{Handling Instance Additions}

Algorithm~\ref{alg:handle-added} describes the procedure executed when a new service instance is introduced in the CC. Upon receiving an event on addition, the LB computes an estimate of the best expected processing latency $l^{p\star}$ based on the $\rho$-percentile of recent processing latency distributions across existing instances of the same service. If the sum of network latency $l_{k,m_{\text{new}}}^n$ and the expected processing latency $l^{p\star}$ does not exceed the QoS threshold~$\tau$, the instance is added to the QoS pool $\mathcal{Q}_k(t)$ with an initial weight $w_{k,m_{\text{new}}}(t)=0$. Upon the next QoS pool maintenance step, this instance will get the best score among the exploration instances because there will be no previous latency records. This ensures that newly added instances are discovered gradually through controlled exploration, avoiding overloading before sufficient feedback is collected.

\begin{algorithm}[]
\caption{Handling instance added}
\label{alg:handle-added}
\begin{algorithmic}[1]
\REQUIRE New instance $m_{\text{new}}$;
\STATE $l^{p\star}\gets$ best expected latency
\IF{$l_{k,m_{\text{new}}}^n+l^{p\star}\le \tau$}
    \STATE \hspace{1em} Add $m_{\text{new}}$ to $\mathcal{Q}_k(t)$; init. weight $w_{k,m_{\text{new}}}(t)\gets 0$
\ENDIF
\end{algorithmic}
\end{algorithm}

\subsection{Handling Instance Removals}

Algorithm~\ref{alg:handle-removed} defines how LBs react to instance removals. When the orchestrator removes an instance or the instance fails, each LB purges its local latency, processing, and cooldown data for that instance, and removes its weight. If the instance was present in the QoS pool $\mathcal{Q}_k(t)$, it is removed and the remaining weights are renormalized to maintain a valid probability distribution over the remaining instances. 

\begin{algorithm}[]
\caption{Handling instance removed}
\label{alg:handle-removed}
\begin{algorithmic}[1]
\REQUIRE Removed instance $m_{\text{rem}}$; QoS pool $\mathcal{Q}_k(t)$ with weights $w_k(t)$
\STATE Remove $w_{k,m}$ from $w_k(t)$
\IF{$m_{\text{rem}} \in \mathcal{Q}_k(t)$}
    \STATE Remove $m_{\text{rem}}$ from $\mathcal{Q}_k(t)$
    \STATE Renormalize remaining weights $w_k(t)$
\ENDIF
\end{algorithmic}
\end{algorithm}

\subsection{Regret Upper Bound}

A regret upper bound represents a theoretical limit on how much cumulative performance loss the learning policy can incur relative to an optimal oracle \cite{slivkins2019}. If the regret grows sublinearly, it means that the average regret per decision $\tfrac{R(T)}{T}$ vanishes as $T$ increases, indicating that the policy asymptotically learns to act optimally. In the context of the CC, this corresponds to load balancers progressively converging toward optimal routing decisions with negligible long-term QoS degradation.

Since the proposed algorithms continuously adapt to time-varying network and load conditions, the CC environment is \emph{non-stationary}: the success probability $\mu_{k,m}(t)$ of each instance evolves over time due to changing latencies, workloads, or instance placements.  
To model this drift formally, we adopt the variation-budget framework proposed by Besbes \textit{et al.}~\cite{besbes2014}.

\begin{definition}[Variation budget per load balancer]
\label{def:variation}
For each load balancer~$k$, define the total variation of instance means over horizon~$T$ as
\begin{equation}
V_k(T)
=\sum_{t=1}^{T-1}
\max_{m\in\mathcal{M}(t)\cap\mathcal{M}(t{+}1)}
|\mu_{k,m}(t{+}1)-\mu_{k,m}(t)|.
\end{equation}
\end{definition}

The value $V_k(T)$ measures the cumulative drift of instance performance seen by load balancer~$k$; a smaller $V_k(T)$ implies a more stable environment.

\begin{theorem}[Regret bound per LB]
\label{thm:nonstationary}
Given the bounded variation from Definition~\ref{def:variation}, the dynamic regret of load balancer~$k$ satisfies
\begin{equation}
\mathbb{E}[R_k(T)]
=\tilde{O}\!\big((M\,V_k(T))^{1/3}\,T^{2/3}\big)
\end{equation}
matching the minimax-optimal upper bound for stochastic multi-armed bandits 
with non-stationary rewards established by~\cite[Thm.~2]{besbes2014}.
\end{theorem}

\begin{proof}
We show that the routing problem at each load balancer $k$ can be cast as a 
stochastic multi-armed bandit with non-stationary rewards in the sense of 
Besbes et al.~\cite{besbes2014}. This allows us to directly apply their 
minimax regret bound for the variation-limited setting.

\textbf{1) Mapping our model to the non-stationary MAB framework.}
Let $\mu_{k,m}(t)$ denote the expected QoS success probability when LB~$k$ routes a request to instance $m$ at time $t$. Since network delay and processing latency vary with congestion and routing conditions, the sequence $\{\mu_{k,m}(t)\}_{t=1}^T$ is generally non-stationary. 
Definition~\ref{def:variation} introduces the total variation budget $V_k(T)$ in exactly the same form as the variation measure used in \cite{besbes2014}, thereby placing our routing problem within the standard variation-limited stochastic MAB framework.

\textbf{2) Our algorithm as a phased/resetting strategy.}
Besbes et al.\ prove that, under bounded variation $V_k(T)$, a near-minimax 
optimal regret of
$\tilde{O}\!\big((M\,V_k(T))^{1/3}T^{2/3}\big)$
is achievable by algorithms combining (i) sliding-window empirical estimation and 
(ii) periodic or event-triggered resets that discard stale observations.
Our decentralized routing algorithm satisfies these conditions:
\begin{itemize}
    \item Sliding-window KDE estimation ensures that mean rewards depend only on recent QoS observations.
    \item The adaptive $\varepsilon$-reset mechanism (triggered by QoS degradation) serves as a data-driven resetting rule, functionally equivalent to the phased/resetting strategies analyzed in~\cite{besbes2014}.
\end{itemize}
\end{proof}

\begin{corollary}[System-wide regret]
\label{cor:system_regret}
Using the definition in Eq.~\ref{eq:overall-regret}, we derive the system-wide regret
\begin{equation}
\mathbb{E}[R(T)]
=\sum_{k=1}^{K}\mathbb{E}[R_k(T)]
=\tilde{O}\!\left(
T^{2/3}\!\sum_{k=1}^{K}(M\,V_k(T))^{1/3}
\right).
\end{equation}
\end{corollary}

The distributed system thus maintains sublinear cumulative QoS loss as long as $V_k(T)=o(T)$ for all load balancers.

\begin{remark}
The regret bound has several important implications:
\begin{itemize}
    \item[(i)] \textbf{Stationary case:} When $V_k(T)=0$, the environment is stationary, i.e. network latencies and client loads remain constant, and the bound is determined by $\tilde{O}(\log T)$, as proven by Auer \textit{et al.}~\cite{auer2002} for the \emph{$\varepsilon$-decay} policy.
    \item[(ii)] \textbf{Sublinearity:} 
    If $V_k(T)=o(T)$, the cumulative drift of instance performance grows slower than linearly with time, resulting in sublinear regret. This situation is typical for practical CC environments, where operating conditions evolve gradually: network latencies change smoothly due to stable routing and link characteristics; client load and processing delays fluctuate continuously but without sudden jumps; and instance placement changes occur on a much slower timescale than request arrivals. Such moderate non-stationarity allows each load balancer to track the optimal routing strategy effectively. 
    On the other hand, when $V_k(T)=\Theta(T)$, the environment changes as quickly as it is observed, making it impossible for any learning-based policy to adapt, and resulting in linear regret. For example, this can happen in highly dynamic mobile edge networks where network paths or client loads vary unpredictably at every step.
    \item[(iii)] \textbf{Estimator impact:} The KDE-based latency estimator improves the stability of empirical QoS estimates by smoothing latency observations, which may affect constant and logarithmic factors in the regret bound, but it does not change the overall growth rate of regret.
\end{itemize}
\end{remark}

\subsection{Computational Complexity and Communication Overhead}
\label{section:complexity}

Each load balancer executes its weight update process independently using only local latency and reward observations, without coordination with other nodes. The computational complexity per decision step is $O(|\mathcal{Q}_k(t)|)$, where $\mathcal{Q}_k(t)$ is the set of instances that could satisfy the QoS requirements, i.e. the QoS pool. All required operations—QoS estimation, pool updates, and weight normalization—scale linearly with the number of candidate instances in the QoS pool.

Communication overhead is minimal, as LBs receive placement and health updates only for nodes within their potential QoS reach. Thus, the proposed algorithm remains lightweight and scalable, 
achieving fully decentralized operation even in large-scale CC environments.

\section{Kubernetes-Native Implementation}
\label{section:impl-qedgeproxy}

This section presents the Kubernetes-native implementation of QEdgeProxy. QEdgeProxy maintains per-service QoS pools and adaptive routing weights, performs request-level scheduling using Smooth Weighted Round Robin (SWRR), and reacts to instance additions, removals, and QoS degradation through an event-driven cooldown mechanism. Since built-in Kubernetes mechanisms (such as NodePort, NetworkPolicy locality options, or service meshes like Istio) do not support adaptive or QoS-driven routing, we deploy QEdgeProxy as a \textit{DaemonSet}, ensuring that every node in the cluster runs its own instance.

\subsection{Deployment as a DaemonSet}

QEdgeProxy is deployed like any other Kubernetes workload, as illustrated in Fig.~\ref{fig:ep-k8s}. Implemented as a lightweight Golang HTTP server,\footnote{Source code: \url{https://github.com/AIoTwin/qedgeproxy/tree/tsc-paper}} each LB connects to the Kubernetes API server to discover the nodes and service instances in the CC. 
To remain aware of runtime changes, each LB subscribes to node- and service-related events. To support this, we extend the standard Kubernetes control plane components with a custom module called the \textit{MQTT Informer}. Implemented in Golang,\footnote{Source code: \url{https://github.com/AIoTwin/k8s-mqtt-informer}} this extension provides event-driven notifications for relevant cluster changes. The MQTT Informer enables lightweight MQTT-based subscriptions to events such as instance creation, removal, or node status updates for specific services. This mechanism directly supports the assumption introduced in Section~\ref{section:problem} that load balancers are informed of instance placement changes but only for nodes within their potential QoS reach.
Clients access services by sending HTTP requests to their assigned QEdgeProxy, with the service name included in the request header. Once an instance is selected, the LB forwards the request directly to the corresponding pod via \textit{kube-proxy}.

To keep the network information up to date, each LB periodically performs lightweight latency measurements by pinging other QEdgeProxy instances. These measured latencies are stored locally and incorporated into QoS pool maintenance.

\begin{figure}[ht]
\centering
\includegraphics[width=0.95\linewidth]{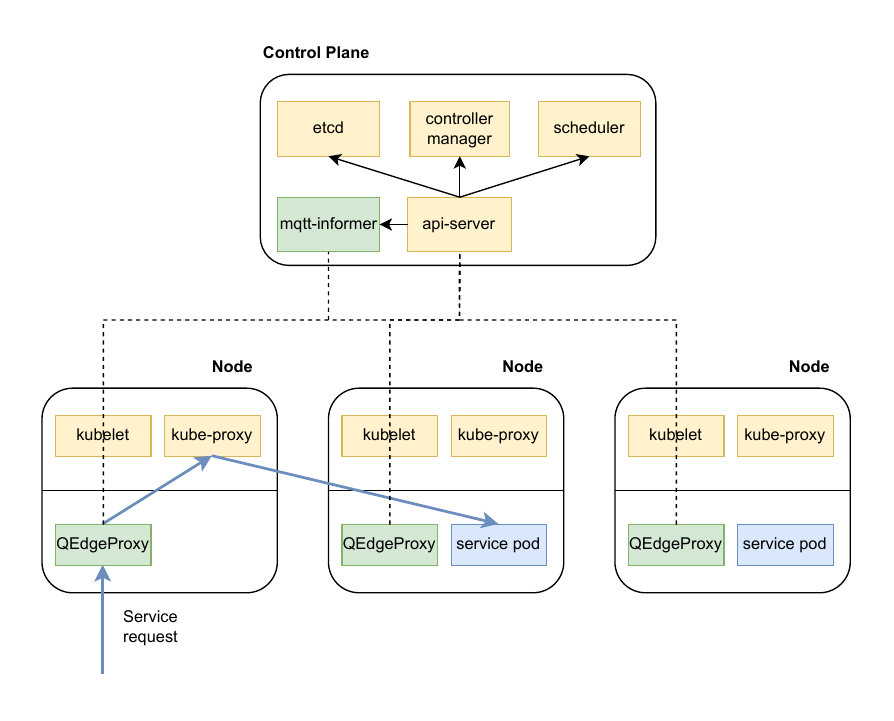}
\caption{QEdgeProxy within a Kubernetes environment.}
\label{fig:ep-k8s}
\end{figure}

\subsection{Runtime Adaptation and Event Handling}

Each QEdgeProxy implements the decentralized algorithms from Section~\ref{section:algorithms} in an event-driven module. For every service~$S$, it maintains a QoS pool~$\mathcal{Q}_k(t)$ containing reachable and healthy instances, along with the adaptive routing weights $w_k(t)$. These weights are continuously updated in the background using observations collected within the sliding time window~$W$, combining measured latency and QoS success feedback.

Each QEdgeProxy also maintains per-instance metrics, including latency, success ratio, error count, and cooldown state. A lightweight periodic maintenance routine updates the QoS pool and recalculates routing weights. Newly discovered or recovered instances are added automatically, while degraded instances are temporarily excluded through a cooldown mechanism. Adaptation parameters such as the maintenance interval~$H_d$, exploration decay factor~$\gamma$, smoothing factor~$\eta$, and cooldown duration~$\Delta_{\mathrm{cd}}$ can be configured externally based on the workload characteristics.

\subsection{Integration and Configuration}

Integrating QEdgeProxy into an existing Kubernetes deployment requires minimal configuration. It works natively with standard Kubernetes services and does not require changes to application code or service definitions. To enable QoS-aware routing, a developer or a DevOps engineer only needs to label the target service with its QoS requirements, i.e., the latency threshold~$\tau$, the success ratio~$\rho$, and the evaluation window $W$. Once an instance of QEdgeProxy is deployed, it automatically discovers the labeled services and their instances (pods), initializes the QoS pools, and continuously performs adaptive request routing according to the defined QoS requirements. No additional configuration files or scaling is required: the proxy autonomously performs monitoring, learning, and routing in the background.

\section{Evaluation}
\label{section:evaluation}

This section evaluates the proposed \textit{QEdgeProxy} load balancing mechanism and its Kubernetes-native implementation using a latency-sensitive IoT edge-AI workload deployed on an emulated CC environment. 

\subsection{Experimental Setup}

Table~\ref{tab:exp-setup} summarizes the configuration used in the evaluation. All experiments were executed on a 10-node \texttt{K3s}\footnote{K3s is a lightweight, resource-efficient Kubernetes distribution designed for edge and IoT environments~\cite{k3s}.} cluster hosting 30 emulated CC nodes with realistic network latencies derived from geographically distributed round-trip time measurements.

\begin{table}[ht]
\centering
\caption{Experimental configuration for evaluating QoS-aware load-balancing strategies.}
\label{tab:exp-setup}
\begin{tabular}{p{0.3\linewidth} p{0.62\linewidth}}
\hline
\textbf{Parameter} & \textbf{Description} \\
\hline
CC Topology & 30 emulated nodes with realistic network simulation \\
Scenarios & Five distinct randomly generated topologies, i.e., five scenarios. \\ 
Load Balancers & 30 (one per node) \\
IoT Service & Quantized \emph{PilotNet} \cite{pilotnet} exposed via HTTP \\
QoS Requirements & 90\% of requests $<$ 80\,ms over a sliding window of 10\,s ($\tau=80$\,ms, $\rho=0.9$, $W=10$\,s) \\
Service Placement & 10 instances deployed on 10 emulated nodes \\
IoT Clients & 120 clients (4 per LB), 10\,req/s each \\
Strategies & \textit{proxy-mity} \cite{proximity}, \textit{Dec-SARSA} \cite{Dec-SARSA}, and \textit{QEdgeProxy} \\
\hline
\end{tabular}
\end{table}

\subsubsection{CC Topology Emulation}
While the underlying infrastructure consisted of 10 nodes, we emulated a larger CC topology with 30 nodes in total via Kubernetes pods. To emulate a resource-constrained CC environment, each emulated node is assigned fixed computational resources through Kubernetes resource limits.\footnote{https://kubernetes.io/docs/concepts/configuration/manage-resources-containers/} Specifically, every node is provisioned with 1\,CPU core and 2\,GB of RAM, approximating the capabilities of a typical single-board computer, such as a Raspberry Pi. 
The network topology was constructed using a latency matrix that captures the pairwise communication delays between nodes and was kept fixed throughout the experiments. Each emulated node was mapped to a city in Europe, and the latencies were drawn from round-trip time (RTT) measurements. Specifically, we relied on the WonderNetwork global ping statistics,\footnote{\url{https://wondernetwork.com/pings}} which provide geographically distributed latency measurements across multiple cities all over the world. Finally, network latencies between nodes were simulated using \textit{Istio VirtualService} \cite{istio}, which applies configurable delays to links between load balancers and service instances using \textit{Fault Injection}.\footnote{\url{https://istio.io/latest/docs/tasks/traffic-management/fault-injection/}}

\subsubsection{IoT Service}
The evaluated service was based on the \emph{PilotNet} model \cite{pilotnet}, a convolutional neural network (CNN) originally introduced by NVIDIA for end-to-end autonomous driving. PilotNet processes raw camera images to predict steering commands, making it a representative workload for latency-critical edge AI applications. In our experiments, we used a quantized version of the model, as in \cite{apj_pilotnet}, to reduce computational load and enable its deployment across multiple resource-constrained CC nodes. The service was containerized and exposed via an HTTP interface to allow clients to send inference requests as HTTP requests. 

The service was configured with QoS requirements stating that \textbf{at least 90\% of all client requests must be processed with an end-to-end latency below 80 milliseconds, evaluated over a 10-second sliding window} ($\tau=80$\,ms, $\rho=0.9, W=10$\,ms). The $80$\,ms latency threshold was chosen to ensure that inference results are returned before the next request (frame) is sent, as clients are sending 10 req/s, thereby maintaining real-time responsiveness. The $90$\% threshold reflects a realistic constraint, as occasional violations are acceptable in many edge-AI systems that can, for instance, rely on a lighter onboard model during temporary periods of degraded QoS.

\subsubsection{Service Placement}
A total of 10 service instances were deployed across 10 distinct emulated nodes, with a service placement strategy ensuring that every node in the CC had at least one service instance within its QoS-defined latency threshold, using the service idle latency (obtained by benchmark) to approximate the QoS reach. The initial placement was computed using a greedy $k$-center algorithm, which iteratively selects, at each step, the node farthest (in network distance) from the currently chosen centers. This approach is commonly used in the literature to address the service placement problem with respect to network latency, such as in~\cite{guo_placement}.

\subsubsection{IoT Clients}
Client workloads were simulated using a lightweight Go program that emulated data streams from on-vehicle cameras. Each simulated client was implemented as an independent \texttt{goroutine} that periodically transmitted an image to the service every 100\,ms, corresponding to a rate of 10 requests per second. A total of 120 clients generating requests were simulated in the experiment, with 4 clients assigned to each load balancer. This configuration ensured a balanced distribution of workload across the system, maintaining sufficient request volume so that the deployed instances could satisfy the QoS requirements for all clients under effective load balancing, while still creating enough load pressure so that suboptimal routing would lead to QoS violations and instance overloads.

\subsubsection{Load-Balancing Strategies}
We compared three distinct load balancing strategies in our experiments. The first strategy, \textit{Proximity-based routing}, implements the \textit{proxy-mity} algorithm proposed by Fahs and Pierre \cite{proximity} (see Section~\ref{section:related}). Two sub-configurations were evaluated: (i) \textit{proxy-mity 1.0}, which prioritizes proximity and therefore minimizes latency (\(\alpha = 1.0\)); and (ii) \textit{proxy-mity 0.9}, which still favors proximity but incorporates a small degree of load balancing (\(\alpha = 0.9\)). These variants allow us to compare our proposed solution both against the configuration that favors the closest instance, and against one that balances latency with load distribution. 

The second strategy, referred to as \textit{Dec-SARSA}, is based on the differential SARSA algorithm for decentralized task scheduling in fog computing, proposed in~\cite{Dec-SARSA} (see Section~\ref{section:related}). In the original formulation, each fog node operates as an autonomous agent that learns to offload tasks to neighboring nodes in order to minimize deadline violations rather than overall latency, which is in line with our QoS requirements. Each agent receives binary rewards indicating whether the latency constraint (deadline) was met and updates its action values using the SARSA algorithm, allowing it to adaptively balance local processing load and offloading decisions without requiring centralized coordination.
In our adaptation, \textit{Dec-SARSA} is reformulated for the CC routing context, where each load balancer functions as an independent RL agent selecting service instances on a per-request basis and learning its routing policy using the SARSA algorithm. As defined in our problem statement in Section~\ref{section:problem}, load balancers operate in a fully decentralized manner without sharing information with other nodes, thus we approximate each instance’s load through its locally observed processing latency, which serves as an implicit indicator of congestion. The agent’s state representation combines recent processing latency statistics with network proximity (RTT) estimates, and actions correspond to the available service instances. Rewards are assigned based on whether a request meets its latency deadline, directly aligning with the QoS-driven objective of our problem statement. This adaptation preserves the decentralized, self-adaptive, and model-free characteristics of the original SARSA scheduler while making it suitable for decentralized load balancing in the CC environment.

The final strategy, \textit{QEdgeProxy}, is our proposed solution described in Section \ref{section:algorithms} which we configured using an $\varepsilon$-decay factor ($\gamma=0.01$), a smoothing factor ($\eta=0.01$), an error count threshold ($E_t=5$), and a cooldown duration ($\Delta_{\mathrm{cd}}=10$\,s). All strategies were implemented in \texttt{Golang}\footnote{Source code and experimental results: \url{https://github.com/AIoTwin/qedgeproxy/tree/tsc-paper}} and deployed as LBs in the same manner as our Kubernetes-native implementation (see Section \ref{section:impl-qedgeproxy}). This ensures that the comparison of different strategies is fair, as the only factor influencing performance is the routing strategy itself.

\subsubsection{Experimental Scenarios}
To ensure the robustness and generalizability of our results, each strategy was evaluated on five distinct CC topologies. By repeating experiments across multiple random topologies, we reduce the risk of bias introduced by a particular geographic configuration and obtain statistically meaningful comparisons of the evaluated strategies.

\newpage

\subsection{Aggregated Results Across Topologies}

\subsubsection{Client-Level QoS Satisfaction Rate}

Fig.~\ref{fig:qos_boxplot} reports the percentage of clients whose QoS requirements are met, aggregated across the five topologies. Across all scenarios, \textit{proxy-mity} exhibits poor QoS satisfaction, with \textit{proxy-mity 0.9} performing worst and \textit{proxy-mity 1.0} only marginally better, typically satisfying QoS for fewer than 50\% of clients. This behavior results from routing decisions based solely on network proximity, which overloads popular instances and ignores explicit QoS constraints.
\textit{Dec-SARSA} substantially improves performance, satisfying QoS for approximately 70--90\% of clients by incorporating latency deadlines into its reward function. However, it remains relatively slow in terms of adapting to overloads, as the Q-values evolve incrementally through differential updates and thus reflect historical performance more than short-term conditions, particularly for load balancers without local instances.
In contrast, \textit{QEdgeProxy} consistently achieves the highest QoS satisfaction (95--100\% of clients) across all scenarios. By assigning dynamic routing weights based on estimated QoS success probabilities and balancing traffic across a pool of feasible instances using SWRR, it adapts more rapidly and stably to changing network and load conditions.
\begin{figure}[ht]
\centering
\includegraphics[width=0.9\linewidth]{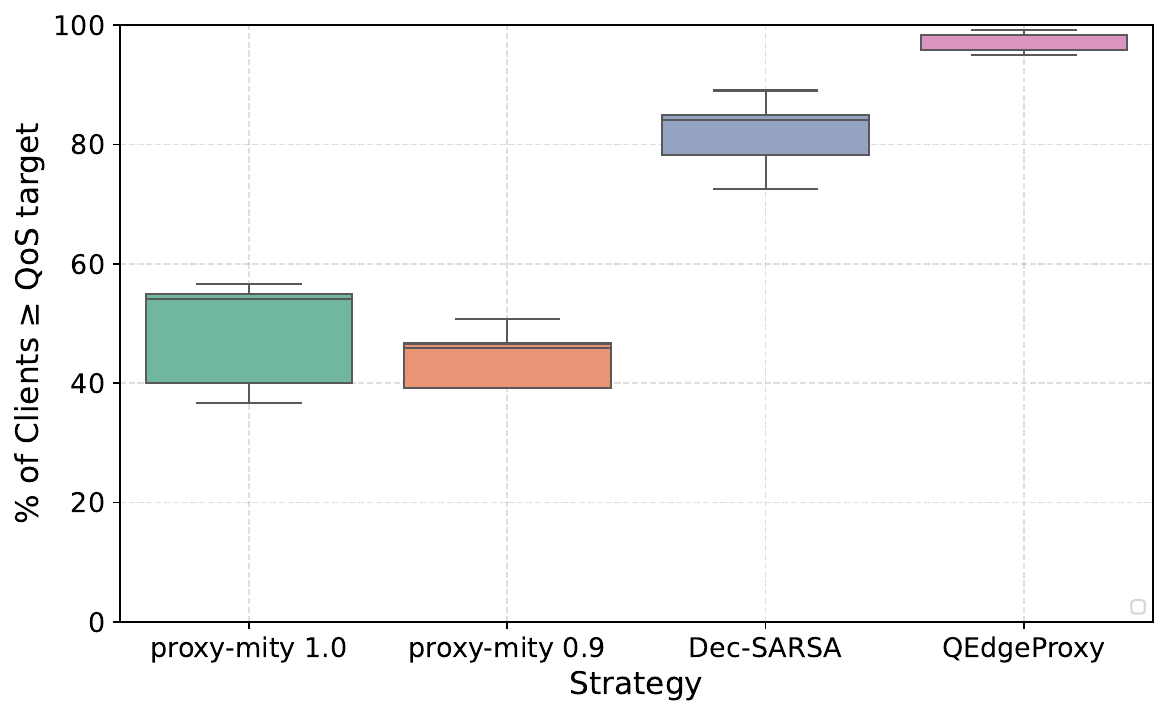}
\caption{Client QoS satisfaction rate aggregated across five scenarios.}
\label{fig:qos_boxplot}
\end{figure}

\subsubsection{Load Balancing Fairness}

We quantify the load balancing fairness using Jain’s fairness index~\cite{jain1984quantitative}, a widely used metric defined as:
\[
F(x_1, \dots, x_n) = \frac{\left(\sum_{i=1}^n x_i\right)^2}{n \cdot \sum_{i=1}^n x_i^2},
\]
where \(x_i\) denotes the load assigned to instance \(i\). The index values range between 0 (worst case) and 1 (perfect fairness). As shown in Fig.~\ref{fig:fairness_boxplot}, the \textit{proxy-mity} strategies exhibit low fairness due to their strong affinity for the nearest instance. \textit{Proxy-mity 0.9} achieves slightly better load balancing than \textit{proxy-mity 1.0}, as it sends 10\% of the requests to instances other than the closest one.  \textit{Dec-SARSA} and \textit{QEdgeProxy} both achieve fairness values of approximately 0.85--0.90. Although both mechanisms rely on performance-driven adaptation, \textit{QEdgeProxy} achieves slightly higher fairness due to its weighted routing strategy. By assigning each instance a weight proportional to its estimated probability of meeting the QoS requirements, it maintains a wider active pool of instances and prevents excessive traffic concentration on a single node, resulting in smoother and more stable load distribution than the slower, per-request adaptation of \textit{Dec-SARSA}.

\begin{figure}[ht]
\centering
\includegraphics[width=0.9\linewidth]{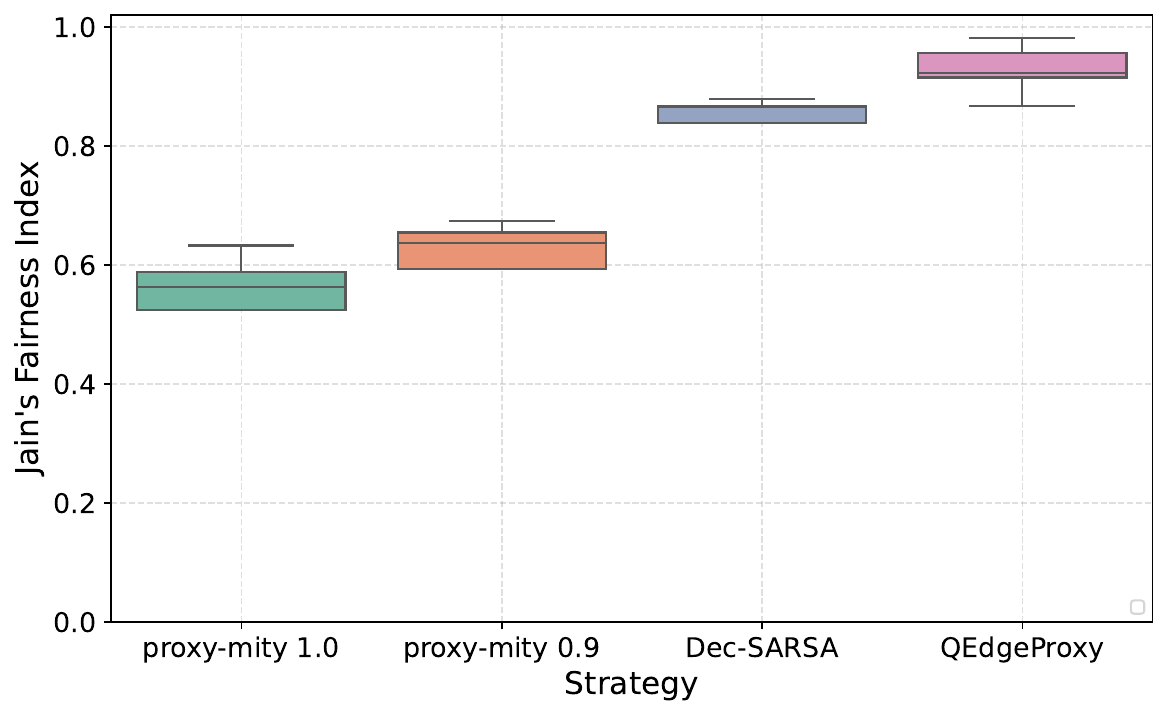}
\caption{Load balancing fairness across five scenarios.}
\label{fig:fairness_boxplot}
\end{figure}

\subsection{Per-Topology Results}

\subsubsection{Per-Client QoS Success}
A representative topology (scenario~\#1) is used to illustrate LB-level behavior. Fig.~\ref{fig:qos_violin} shows the distribution of per-client QoS success rate $\hat{\rho}$. This plot shows the distribution of success ratios among all clients, indicating whether a strategy treats clients uniformly, as defined in the objective of the problem statement, or excessively favors some of them. Client QoS success with both \textit{proxy-mity} configurations spans the full range from 0 to 100\%, reflecting highly uneven behavior. Clients with rates close to 100\% are typically single users of their closest instance, so that instance is never overloaded, while clients with rates near 0\% are the ones that share the closest instance with multiple other clients, causing its overload. Due to occasional offloading to the non-nearest instance, \textit{proxy-mity 0.9} does not achieve 100\% success for any client.
\textit{Dec-SARSA} has significantly better distribution, with most clients exceeding 85\% QoS success, but 15 clients (out of 120) still fall below the target threshold of 90\%, all connected to load balancers without a local instance. In contrast, \textit{QEdgeProxy} fails to meet the target for only three clients, all marginally below the threshold, and exhibits an overall upward shift in the distribution, indicating more uniformly high per-client QoS satisfaction.

\begin{figure}[ht]
\centering
\includegraphics[width=0.9\linewidth]{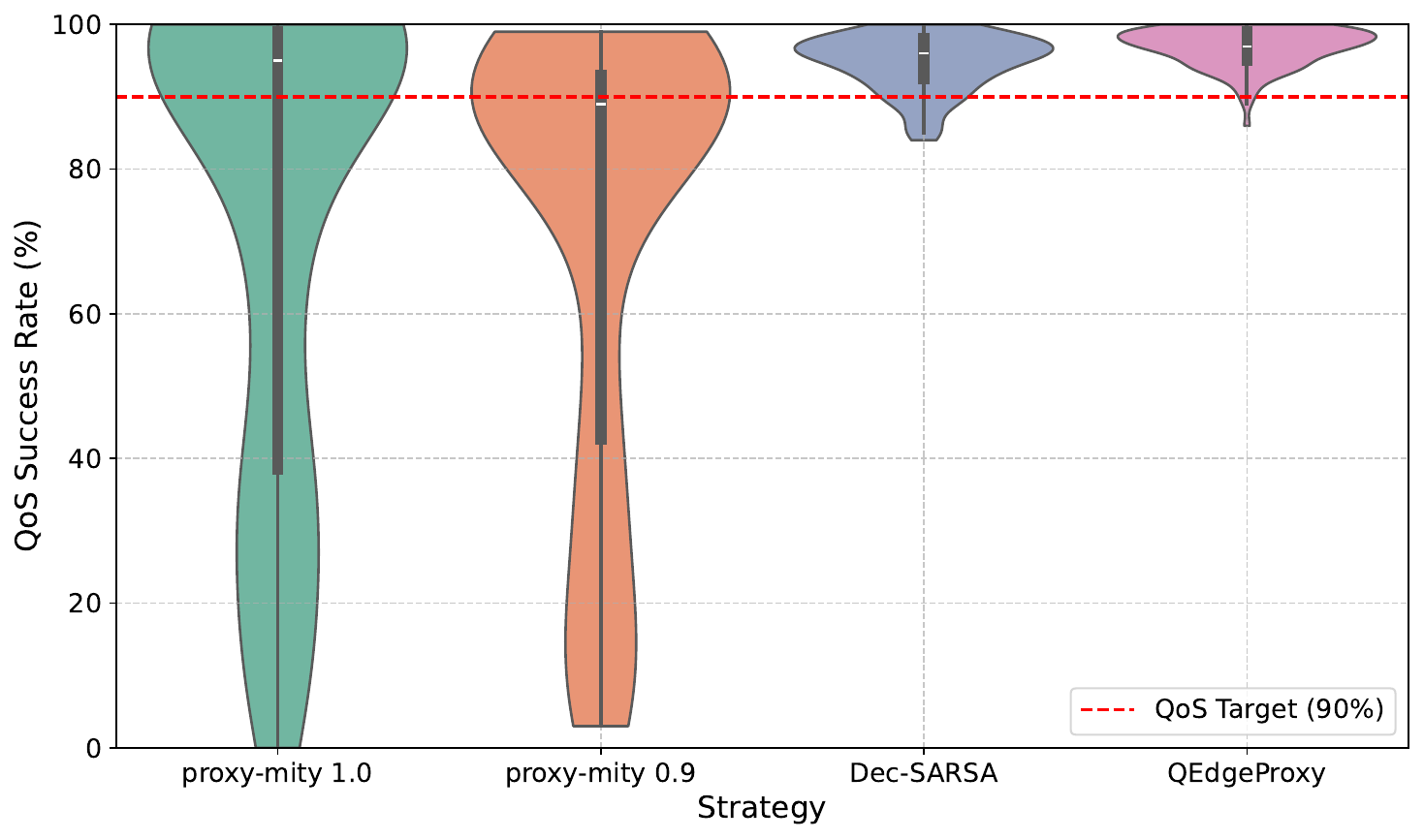}
\caption{Distribution of client QoS success rate $\hat{\rho}$ (scenario~\#1).}
\label{fig:qos_violin}
\end{figure}

\subsubsection{Rolling QoS Over Time}

Fig.~\ref{fig:rolling-qos} illustrates how the QoS success rate evolves over time, showing the rolling average computed over a 10-second window. Both \textit{proxy-mity} configurations quickly converge to a stable rate of 50--60\%, which remains nearly constant throughout the experiment. This occurs because the routing weights are fixed at initialization and never updated, so the QoS success rate depends only on instance load and response time. \textit{Dec-SARSA} converges faster than \textit{QEdgeProxy}, reaching a mostly steady success rate after approximately 30~seconds, compared to about 60~seconds for \textit{QEdgeProxy}. The slower initial convergence of \textit{QEdgeProxy} occurs because: (i) routing weights are updated at discrete intervals (apart from cooldown events), and (ii) a higher initial exploration rate delays early exploitation. However, this additional exploration allows \textit{QEdgeProxy} to better identify instances capable of fulfilling the QoS, leading to superior long-term performance. After convergence, \textit{QEdgeProxy} maintains a consistently higher QoS success rate, whereas \textit{Dec-SARSA} exhibits higher fluctuations due to its per-request routing behavior, where the policy tends to exploit a single instance until it becomes overloaded and then abruptly shifts to another instance. 

\begin{figure}[ht]
\centering
\includegraphics[width=0.9\linewidth]{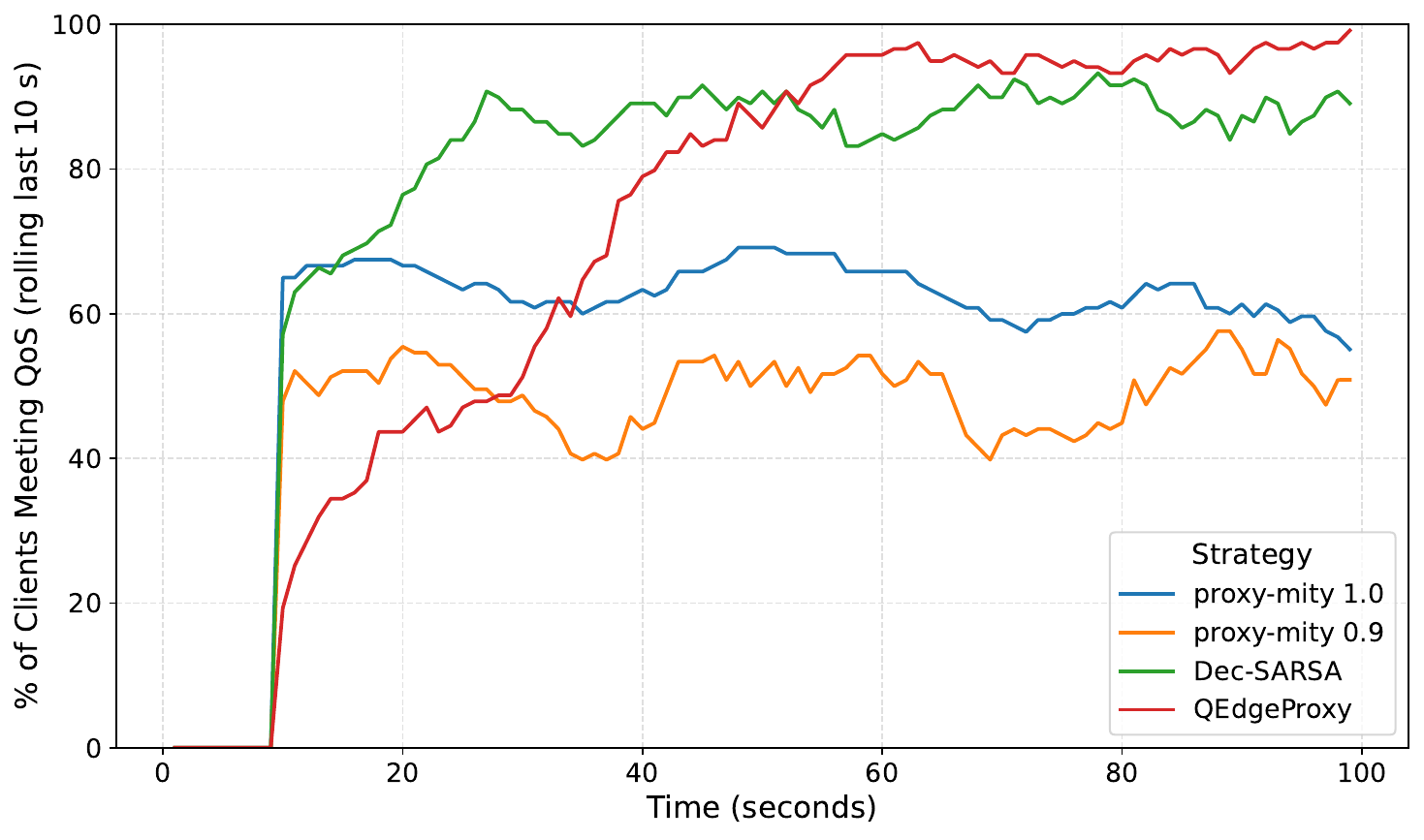}
\caption{Rolling QoS success rate (scenario~\#1).}
\label{fig:rolling-qos}
\end{figure}

\subsubsection{Request Distribution Across Instances}

Fig.~\ref{fig:request-rate} illustrates the request distribution across service instances in scenario~\#1. Both \textit{proxy-mity} variants heavily overload instance~\textit{i9}, with \textit{proxy-mity 1.0} producing the most imbalanced allocation, consistent with the fairness results in Fig.~\ref{fig:fairness_boxplot}.
In contrast, \textit{Dec-SARSA} and \textit{QEdgeProxy} exhibit similar and more balanced request distributions. Instance~\textit{i2} receives a constant 40~req/s under both strategies because it lies outside the QoS reach of all nodes except its host, and is therefore correctly excluded from most routing decisions. Reported values reflect average request rates over the experiment, while short-term load variations are discussed in the processing latency analysis.

\subsubsection{Processing Latency per Instance}

Fig.~\ref{fig:p90-latency} shows the p90 (90th percentile) processing latency per instance in scenario~\#1, highlighting the impact of routing decisions on instance congestion. As expected, processing latency closely follows the request distribution in Fig.~\ref{fig:request-rate}.
Both \textit{proxy-mity} variants overload instances~\textit{i5} and \textit{i9}, resulting in elevated latencies and QoS violations. While \textit{Dec-SARSA} and \textit{QEdgeProxy} achieve more balanced behavior, \textit{Dec-SARSA} exhibits higher latency peaks as it occasionally overloads instances before redistributing requests. In contrast, \textit{QEdgeProxy} maintains lower and more stable processing latencies by proactively distributing load across QoS-feasible instances. Processing latency must be interpreted jointly with network latency, as higher instance latency can still satisfy QoS for clients with sufficiently low network delay.

\begin{figure}[t]
\centering
\begin{minipage}[t]{0.48\columnwidth}
    \centering
    \includegraphics[width=\linewidth]{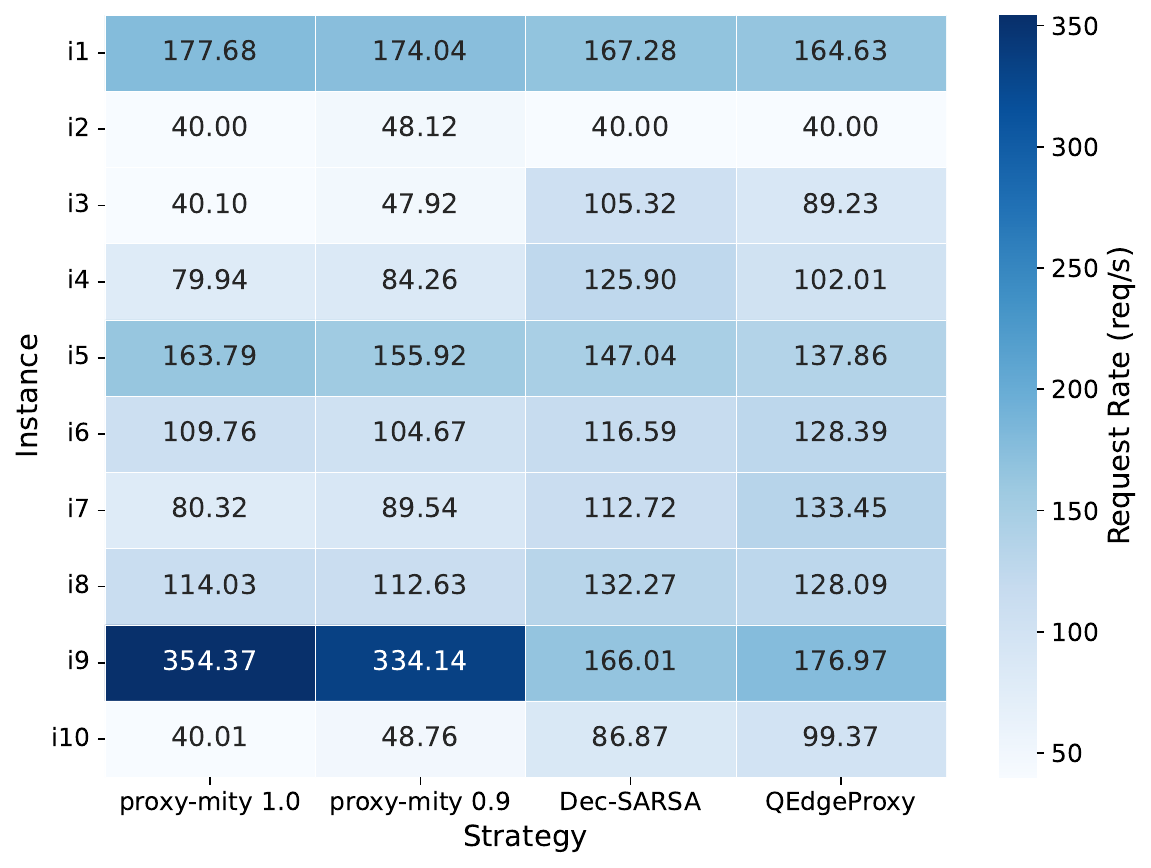}
    \caption{Request distribution across service instances (scenario~\#1).}
    \label{fig:request-rate}
\end{minipage}
\hfill
\begin{minipage}[t]{0.48\columnwidth}
    \centering
    \includegraphics[width=\linewidth]{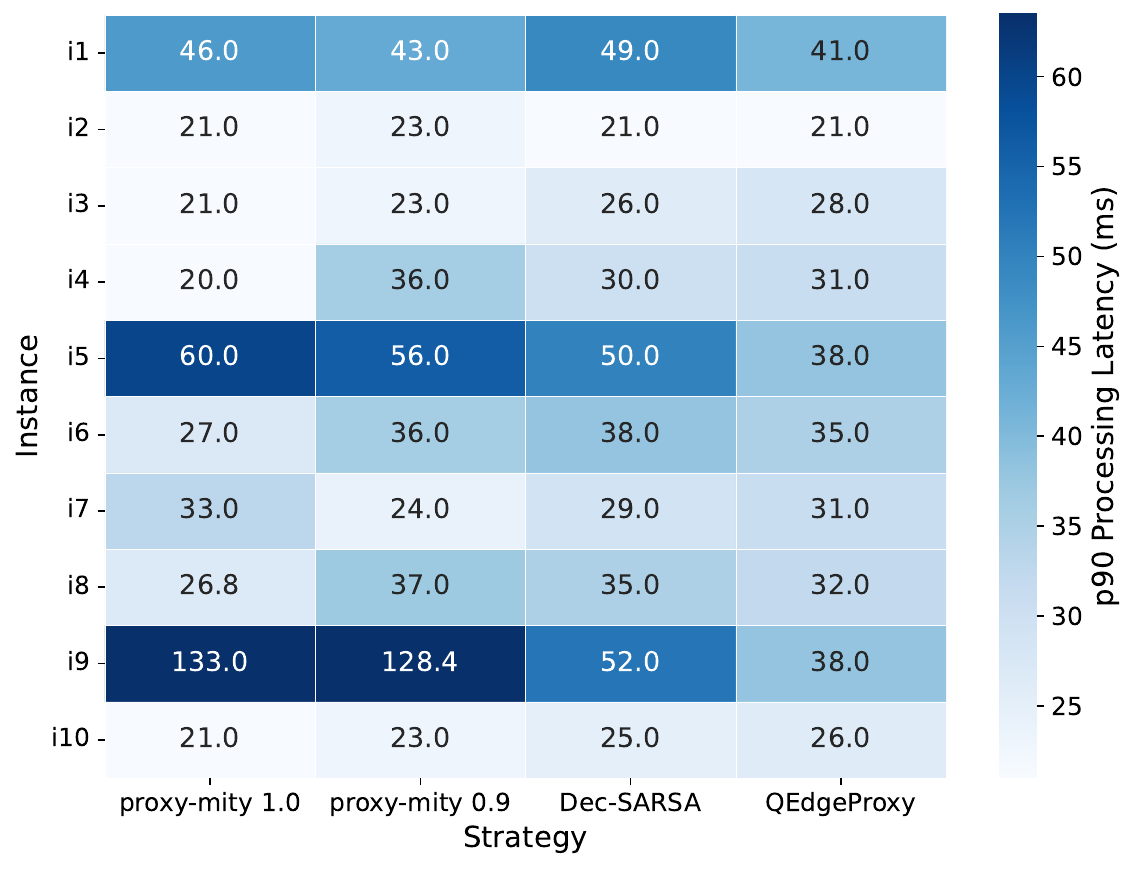}
    \caption{p90 processing latency per instance (scenario~\#1).}
    \label{fig:p90-latency}
\end{minipage}
\end{figure}

\subsubsection{Request Distribution from a Single Load Balancer}

To illustrate per-node routing behavior, Fig.~\ref{fig:requests_heatmap_lb_local} and Fig.~\ref{fig:requests_heatmap_lb_nonlocal} show the request distribution of two representative load balancers in scenario~\#1: one with a local service instance and one without. As expected, \textit{proxy-mity} routes most traffic to the closest instance (with \textit{proxy-mity 0.9} only slightly diversifying), while \textit{Dec-SARSA} mainly exploits the local instance and offloads reactively under overload. In contrast, \textit{QEdgeProxy} maintains a QoS-feasible pool and distributes requests across multiple instances, with a stronger bias toward the local instance when available and a broader pool when no local instance exists.

\begin{figure}[t]
\centering
\subfloat[LB with a local instance.]{
    \includegraphics[width=0.48\columnwidth]{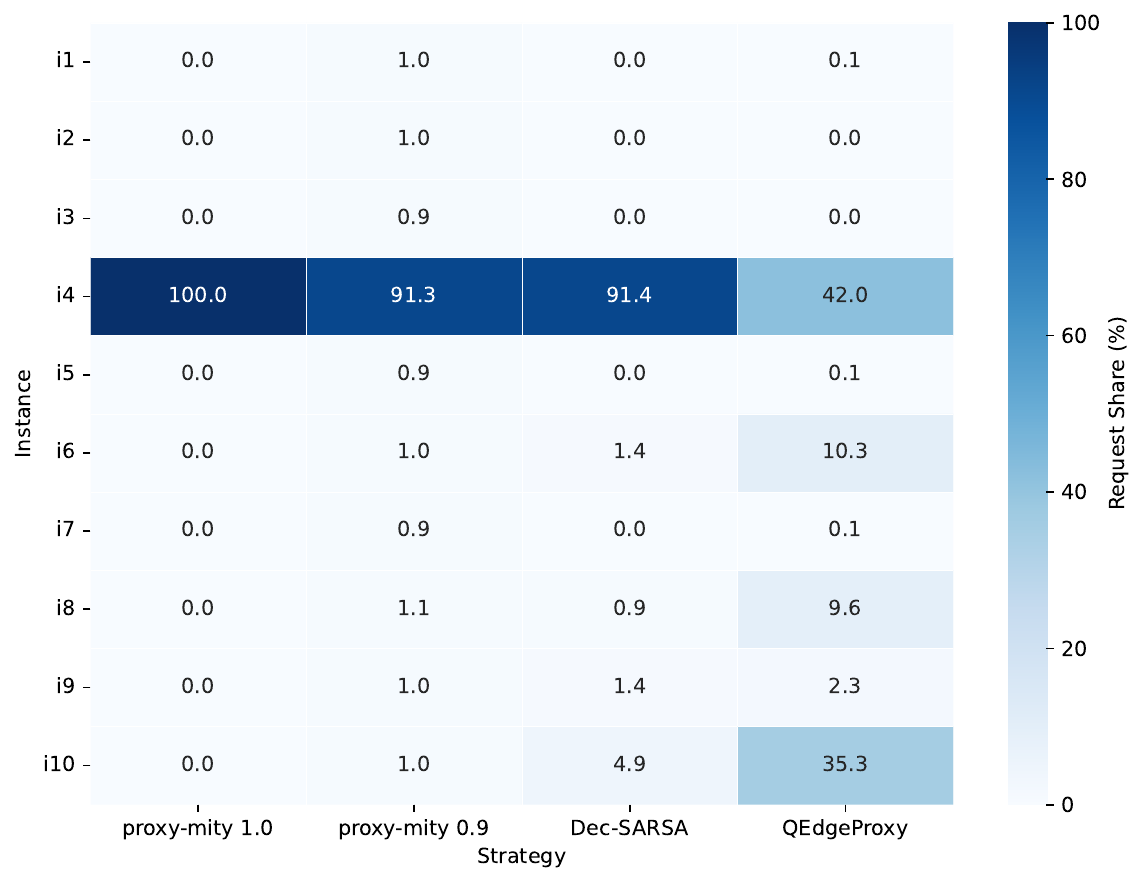}
    \label{fig:requests_heatmap_lb_local}
}
\subfloat[LB without a local instance.]{
    \includegraphics[width=0.48\columnwidth]{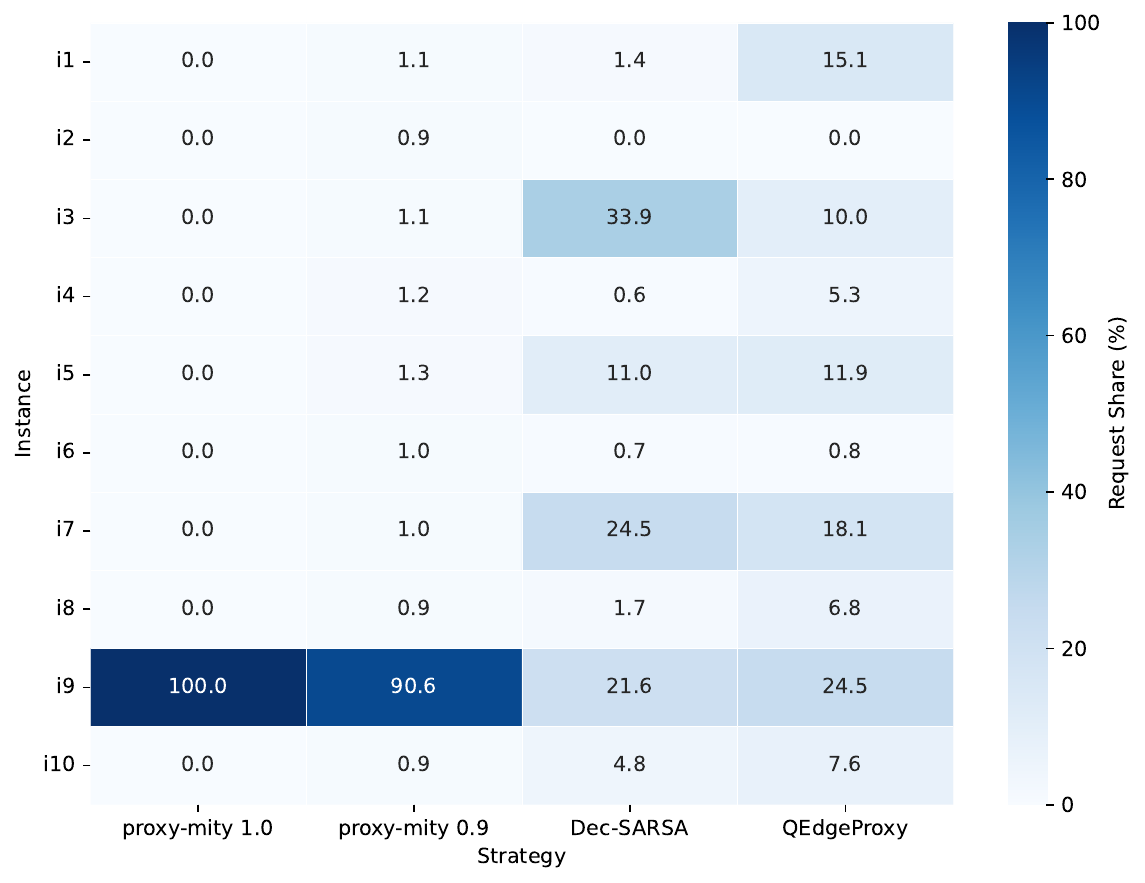}
    \label{fig:requests_heatmap_lb_nonlocal}
}
\caption{Request distribution on a single load balancer (scenario~\#1).}
\label{fig:requests_heatmap_lb}
\end{figure}

\subsection{Adaptiveness}

To evaluate how the proposed solution adapts to dynamic changes in the ECC, we conducted two additional experiments. 

\subsubsection{Client Dynamics}
To evaluate adaptiveness to client dynamics, the system starts with 60 clients (2 per LB) using scenario~\#1, and 30 additional clients are introduced mid-experiment by assigning them to 15 randomly selected LBs, creating localized load imbalances. As shown in Fig.~\ref{fig:adapt_clients}, both \textit{proxy-mity} variants initially achieve near-perfect QoS under low load but degrade sharply after the client load increase due to instance overload, with \textit{proxy-mity 0.9} exhibiting a larger drop. In contrast, \textit{Dec-SARSA} and \textit{QEdgeProxy} rapidly adapt to the increased load and stabilize around 90\% QoS success, enabled by their reactive learning mechanisms that redistribute traffic within a single QoS evaluation window.

\begin{figure}[ht]
\centering
\includegraphics[width=0.9\linewidth]{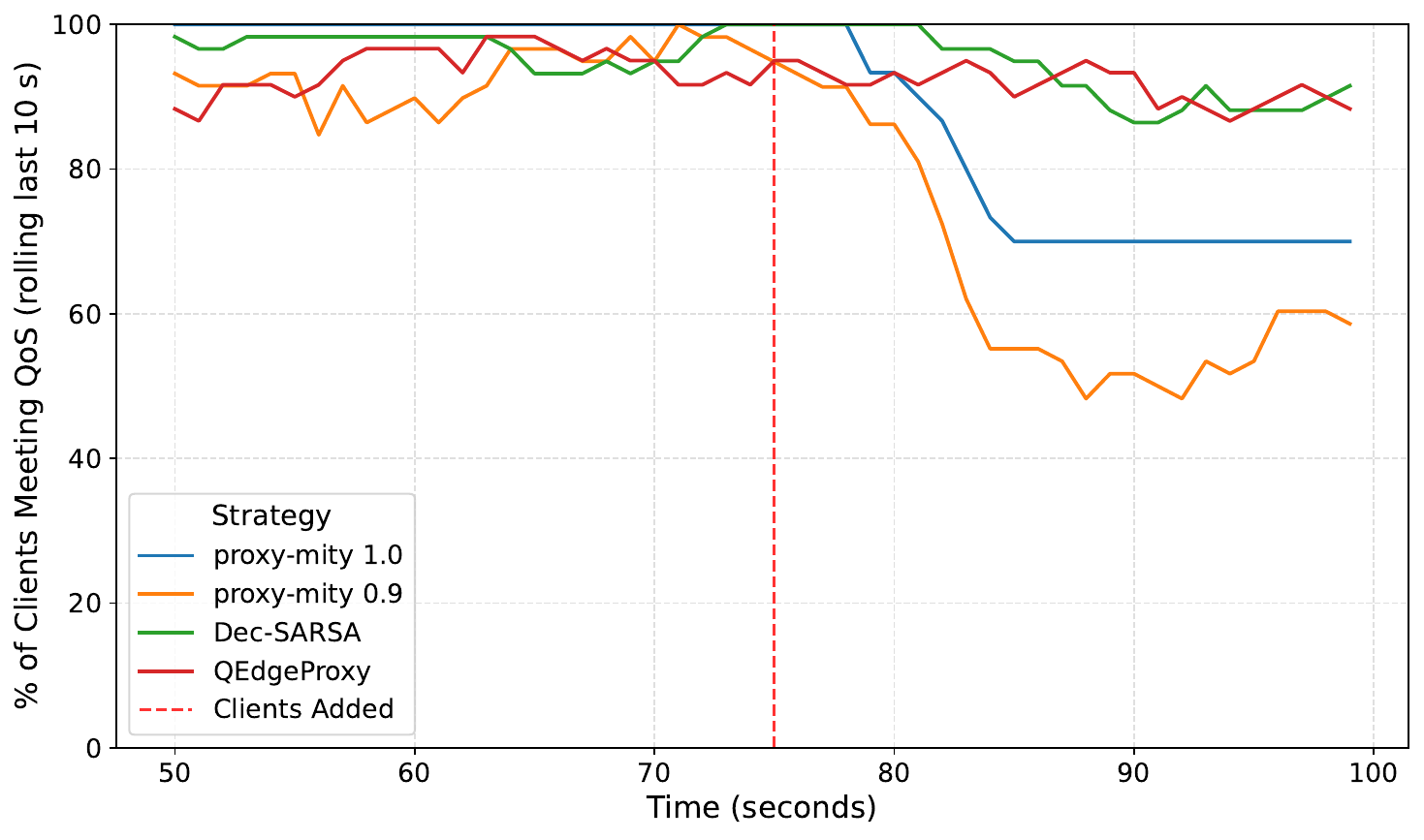}
\caption{Adaptiveness to increased client load.}
\label{fig:adapt_clients}
\end{figure}

\subsubsection{Instance Removal}

In the second adaptiveness experiment, the system starts with 120 clients and 10 instances (scenario~\#1), and one instance (\texttt{i10}) is removed mid-experiment to emulate a sudden capacity reduction. As shown in Fig.~\ref{fig:adapt_instance}, all strategies experience an initial QoS drop after the removal. The \textit{proxy-mity} variants show a limited decline, affecting mainly the clients previously served by the removed instance. Both \textit{Dec-SARSA} and \textit{QEdgeProxy} adapt to the change, but with different outcomes: \textit{Dec-SARSA} stabilizes around 80\% QoS success due to its slower, reward-driven updates, whereas \textit{QEdgeProxy} recovers faster and maintains close to 90\% QoS by dynamically recomputing instance weights and rebalancing its QoS pools.

\begin{figure}[ht]
\centering
\includegraphics[width=0.9\linewidth]{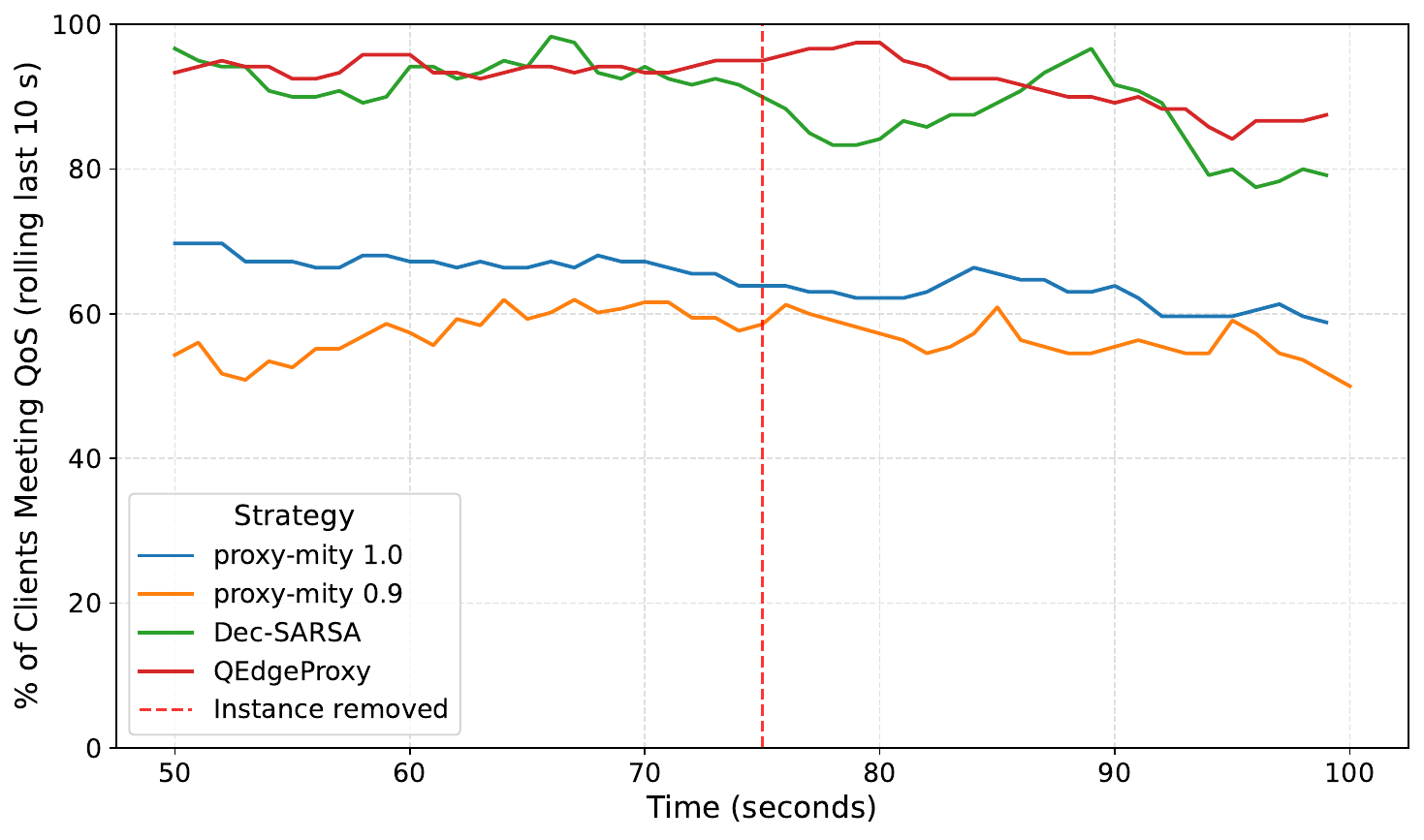}
\caption{Adaptiveness to instance removal.}
\label{fig:adapt_instance}
\end{figure}

\subsection{Resource Footprint}

During the experiment, we recorded multiple snapshots of \textit{QEdgeProxy}'s resource consumption under peak request load (4 clients, $40$~req/s) to evaluate its resource overhead in an operational environment. The results show an average CPU utilization of \textbf{0.13~cores} and memory usage of \textbf{60~MB}, confirming that the proxy can operate efficiently on resource-constrained CC nodes such as the Raspberry~Pi.

\subsection{Discussion and Limitations}
\label{section:limitations}

Several directions may further enhance our proposed load-balancing solution:
\begin{itemize}
    \item \textbf{Multi-dimensional QoS objectives.}  
    The current mechanism targets latency-based QoS requirements. Future extensions could incorporate additional QoS parameters such as reliability, energy consumption, or privacy constraints, enabling instance selection based on multi-objective criteria, such as studied by~\cite{ASLANPOUR2024266}.
    \item \textbf{Service chains and interdependent workloads.}  
    Real-world IoT applications often involve multi-stage pipelines where end-to-end QoS depends on cumulative QoS across several services. Extending QEdgeProxy to account for service chains would significantly broaden its applicability.
    \item \textbf{Collaborative prediction across neighboring LBs.}  
    Exploring mechanisms for neighboring LBs to share local observations or summarized QoS statistics could improve QoS prediction quality, especially in highly dynamic CC environments.
\end{itemize}

\section{Conclusion}
\label{section:conclusion}

This work addressed the challenge of ensuring per-client QoS satisfaction for latency-sensitive IoT applications in the dynamic and heterogeneous Computing Continuum (CC). 
We formulated decentralized load balancing in the CC as a Multi-Player Multi-Armed Bandit (MP-MAB) problem with heterogeneous rewards and proposed \textit{QEdgeProxy}, a lightweight decentralized load balancer (LB) deployed on every CC node.

QEdgeProxy maintains a \textit{QoS pool}---a dynamically updated set of service instances predicted to satisfy the given QoS requirements. 
The QoS pool is constructed using Kernel Density Estimation (KDE) over a sliding window of recent observations to estimate instance-level QoS success probabilities. 
Routing decisions combine these estimates with an adaptive $\varepsilon$-decay exploration scheme that dynamically adjusts in response to observed QoS satisfaction, and Smooth Weighted Round Robin (SWRR) routing, enabling each LB to balance exploitation of high-quality instances with controlled exploration of alternatives. 
This design allows QEdgeProxy to avoid instance overloads, quickly identify improving or newly deployed instances, and deliver stable performance under time-varying network and workload conditions.

We implemented the proposed mechanism as a Kubernetes-native component and evaluated it on a real distributed testbed comprising a 10-node K3s cluster emulating a 30-node CC topology with realistic latency distributions. 
Using a latency-sensitive edge-AI workload and comparing to the proximity-based and reinforcement learning baselines, the evaluation demonstrated that QEdgeProxy:
\begin{itemize}
    \item achieves the highest per-client QoS satisfaction across all tested topologies (95--100\% of clients meeting the target),
    \item provides more balanced request distribution than given baselines,
    \item adapts robustly to dynamic changes, such as client-load surges and instance removals, with faster recovery and more stable convergence.
\end{itemize}

These results show that decentralized MP-MAB-based load balancing with KDE-driven QoS estimation is a practical and effective foundation for QoS-aware load balancing in the CC, particularly in environments where global coordination is costly or impractical.

\bibliographystyle{IEEEtran}
\bibliography{IEEEabrv,references.bib}

\vspace{-1cm}

\begin{IEEEbiography}[{\includegraphics[width=1in,height=1.25in,clip,keepaspectratio]{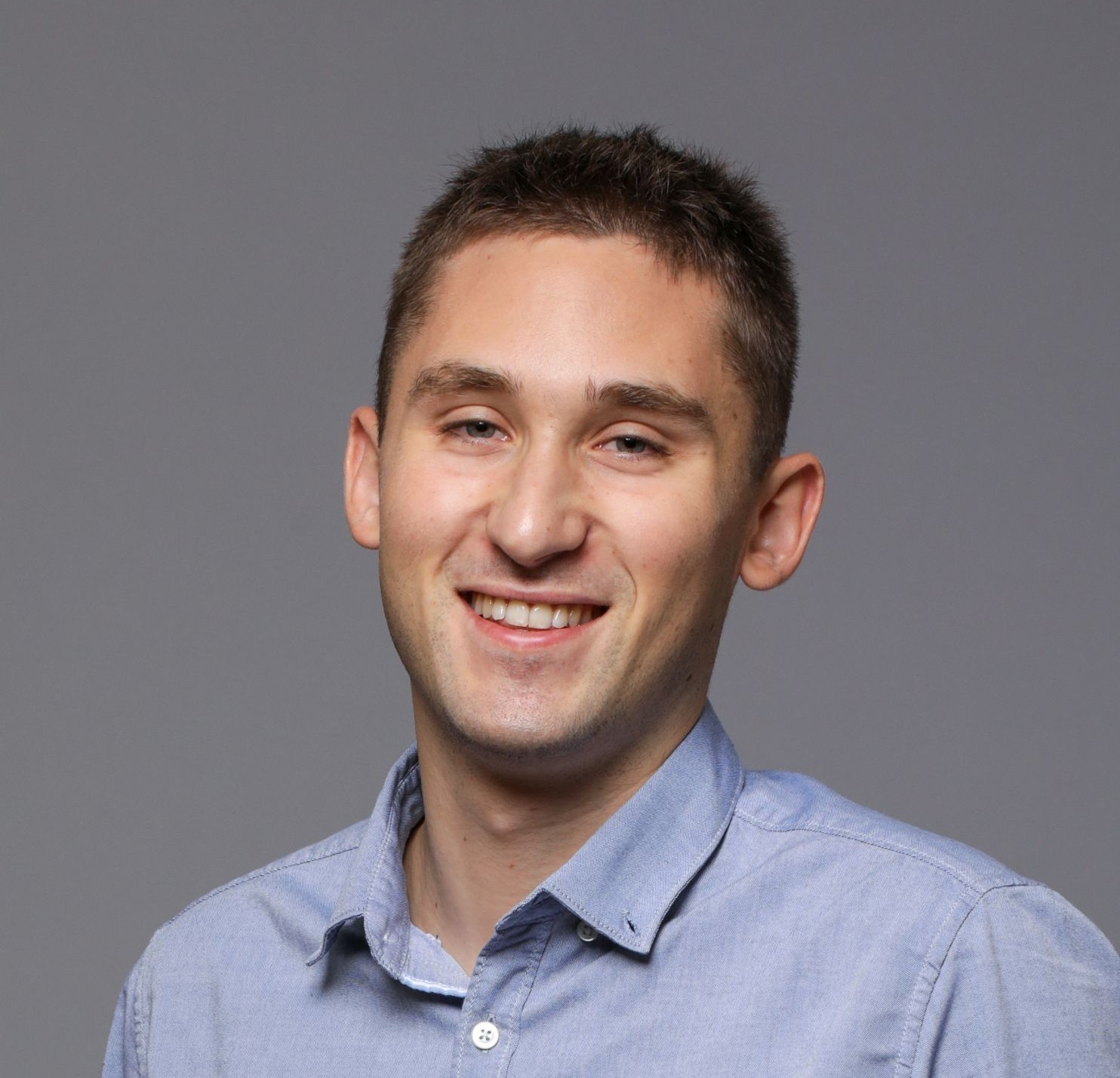}}]{Ivan Čilić}
received the M.Sc. degree in information and communication technology from the Faculty of Electrical Engineering and Computing, University of Zagreb, Croatia, in 2020, where he is currently pursuing the Ph.D. degree in computing. He is a Leading Researcher with the Department of Telecommunications at the same faculty. His research interests include edge computing, service orchestration, load balancing and federated learning.
\end{IEEEbiography}

\vspace{-1cm}

\begin{IEEEbiography}[{\includegraphics[width=1in,height=1.25in,clip,keepaspectratio]{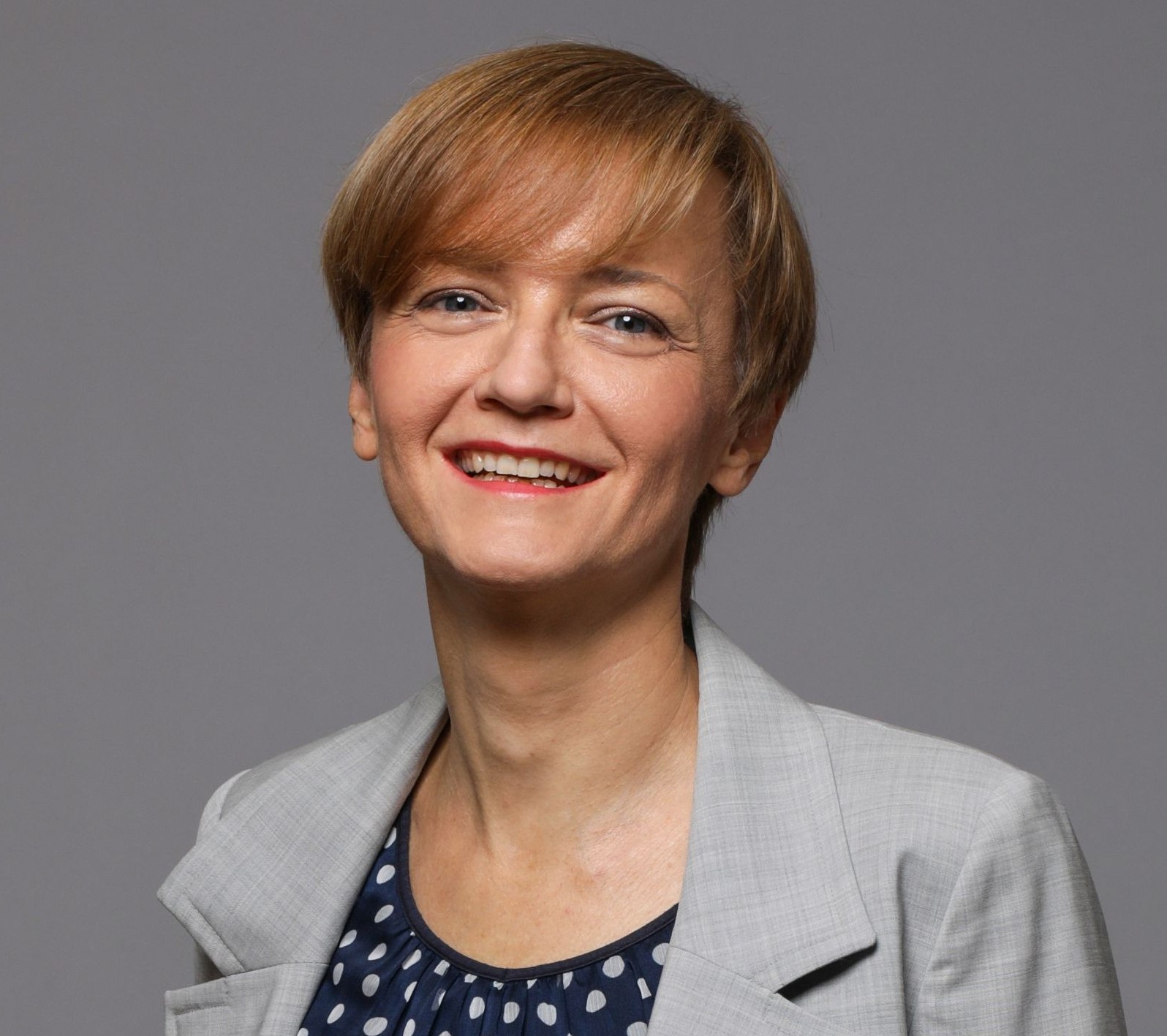}}]{Ivana Podnar Žarko}
is a Full Professor at the Faculty of Electrical Engineering and Computing, University of Zagreb, Croatia, where she leads the Internet of Things Laboratory. She has participated in numerous national and European research projects, including serving as Technical Manager of the H2020 project symbIoTe, and currently coordinates the Horizon Europe project AIoTwin. Her research interests include large-scale distributed systems, IoT interoperability, and edge computing. She has co-authored over one hundred scientific publications and has served on program committees and editorial boards of several international conferences and journals.
\end{IEEEbiography}

\vspace{-1cm}

\begin{IEEEbiography}[{\includegraphics[width=1in,height=1.25in,clip,keepaspectratio]{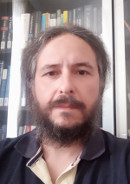}}]{Pantelis A. Frangoudis}
is a researcher with the Distributed Systems Group, TU Wien, Austria. He has been with the Communication Systems Department, EURECOM, France (2017–2019), and with team DIONYSOS at  RISA/INRIA Rennes, France (2012–2017), which he originally joined under an ERCIM “Alain Bensoussan” postdoctoral fellowship. He has a Ph.D. (2012) in Computer Science from AUEB, Greece. His interests include mobile and wireless networking, edge and cloud computing, and Internet multimedia.
\end{IEEEbiography}

\vspace{-1cm}

\begin{IEEEbiography}[{\includegraphics[width=1in,height=1.25in,clip,keepaspectratio]{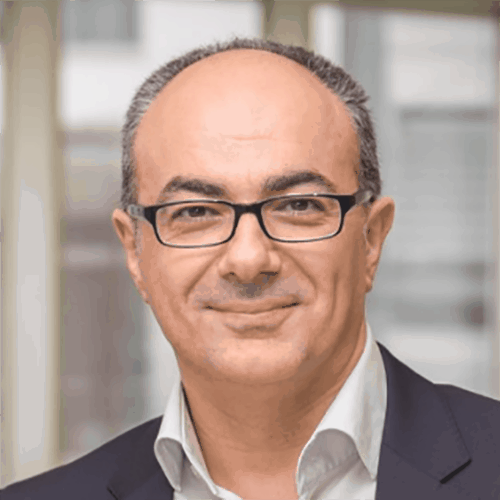}}]{Schahram Dustdar}
(Fellow, IEEE) is a Full Professor of Computer Science (Informatics) with a focus on Internet Technologies, heading the Distributed Systems Group at the TU Wien. He is the founding co-Editor-in-Chief of ACM Transactions on Internet of Things (ACM TIoT) as well as Editor-in-Chief of Computing (Springer). He is an Associate Editor of IEEE Transactions on Services Computing, IEEE Transactions on Cloud Computing, ACM Computing Surveys, ACM Transactions on the Web, and ACM Transactions on Internet Technology, as well as on the editorial board of IEEE Internet Computing and IEEE Computer. He is a recipient of multiple awards: TCI Distinguished Service Award (2021), IEEE TCSVC Outstanding Leadership Award (2018), IEEE TCSC Award for Excellence in Scalable Computing (2019), ACM Distinguished Scientist (2009), ACM Distinguished Speaker (2021), IBM Faculty Award (2012).
\end{IEEEbiography}

\newpage

\vfill

\end{document}